\documentclass[acmsmall]{acmart}
 
\AtBeginDocument{%
  \providecommand\BibTeX{{%
    \normalfont B\kern-0.5em{\scshape i\kern-0.25em b}\kern-0.8em\TeX}}}

\usepackage{helvet}
\usepackage{courier}
\usepackage{graphicx}
\usepackage{amsmath}
\usepackage{multirow}
\usepackage{booktabs}
\usepackage{xcolor}
\usepackage{bm}
\usepackage{mathrsfs}
\usepackage{amsthm,amssymb}
\usepackage{subfigure}
\usepackage{lineno}
\usepackage{color}
\usepackage[ruled]{algorithm2e}
\newtheorem{problem}{Problem}
\newtheorem{definition}{Definition}
\newtheorem{theorem}{Theorem}

\newcommand{\squishlist}{
   \begin{list}{$\bullet$}
    { \setlength{\itemsep}{0pt}      \setlength{\parsep}{2pt}
      \setlength{\topsep}{1pt}       \setlength{\partopsep}{0pt}
      \setlength{\leftmargin}{1em} \setlength{\labelwidth}{1em}
      \setlength{\labelsep}{0.5em} } }

\newcommand{\squishend}{
    \end{list}  }
\newcommand{\topcaption}{%
\setlength{\abovecaptionskip}{5pt}%
\setlength{\belowcaptionskip}{5pt}%
\caption}

\begin{document}

\title{Multi-Stage Network Embedding for Exploring Heterogeneous Edges}

\author{Hong Huang}
    \email{honghuang@hust.edu.cn}
\author{Yu Song}
    \email{yusonghust@hust.edu.cn}
\author{Fanghua Ye}
    \email{smartyfh@outlook.com}
\author{Xing Xie}
    \email{xing.xie@microsoft.com}
\author{Xuanhua Shi}
    \email{xhshi@hust.edu.cn}
\author{Hai Jin}
    \email{hjin@hust.edu.cn}
    
\thanks{Hong Huang, Yu Song, Xuanhua Shi and Hai Jin are with the National Engineering Research Center
for Big Data Technology, Service Computing Technology
and System Lab, and School of Computer Science and Technology, Huazhong University of Science and Technology, Wuhan, 430074, China. Fanghua Ye is with the Department of Computer Science, University College London, London, UK. Xing Xie is with Microsoft Research Asia, Beijing, China.}

\renewcommand{\shortauthors}{}
\begin{abstract}
 The relationships between objects in a network are typically diverse and complex, leading to the heterogeneous edges with different semantic information. In this paper, we focus on exploring the heterogeneous edges for network representation learning. By considering each relationship as a view that depicts a specific type of proximity between nodes, we propose a multi-stage non-negative matrix factorization (MNMF) model, committed to utilizing abundant information in multiple views to learn robust network representations. In fact, most existing network embedding methods are closely related to implicitly factorizing the complex proximity matrix. However, the approximation error is usually quite large, since a single low-rank matrix is insufficient to capture the original information. Through a multi-stage matrix factorization process motivated by gradient boosting, our MNMF model achieves lower approximation error. Meanwhile, the multi-stage structure of MNMF gives the feasibility of designing two kinds of non-negative matrix factorization (NMF) manners to preserve network information better. The united NMF aims to preserve the consensus information between different views, and the independent NMF aims to preserve unique information of each view. Concrete experimental results on realistic datasets indicate that our model outperforms three types of baselines in practical applications.
\end{abstract}




\maketitle

\section{Introduction}\label{introduction}
The relationships among objects typically take the form of networks in the real world, such as biological networks~\cite{li2017scored}, social networks~\cite{van2014online}, and academic networks~\cite{tang2008arnetminer}. Raw network data are usually complex and hard to be processed. Traditional network representation methods, such as adjacent matrix and incidence matrix, often bring extra challenges for the subsequent social computing tasks \cite{cui2018survey}. Recently, network representation learning~(a.k.a. network embedding), aiming to transform networks into low dimensional vectors, has aroused widespread concerns from industry and academia. Meanwhile, real-world networks usually contain abundant information, which further leads to complex and heterogeneous edges between nodes in networks. For example, in an academic network, there may be academic collaboration between the two authors, meanwhile, one author can also cite the papers by another author, causing the relationships between two authors to be complicated and diverse. Another example is the twitter network, one user can retweet, reply, like and mention another user's tweets. All of these behaviors can be abstracted as an edge between two users to reflect their interactions but different behaviors indicate different intimacy between two users. Obviously, different types of relationships play an important role in exploring the heterogeneous edges for many emerging applications. To this end, it is quite reasonable to embed such networks into a low dimensional vectors via considering the semantic meanings of heterogeneous edges.    

Heterogeneous networks are notoriously difficult to mine, thus a great deal of effort has been devoted to representing networks with heterogeneous edges. In the last few years, a plethora of heterogeneous information network (HIN) embedding approaches~\cite{tang2015pte,dong2017metapath2vec,fu2017hin2vec,lu2019relation} have been successfully applied to explore the heterogeneity of network data. However, these methods typically fuse the information of multiple types of relationships into an unified representation without considering the relation-specific representations, which ignores the nuanced differences between these relationships. For instance, two students share a common classmate and two students both join a same league. Hence, two types of edges are established between the two students reflecting two types of relationships. It's worth noting that both edges can have the same meaning, that is, two students are similar in some aspects. But the former edge may indicate the educational information of them while the latter may indicate the information about their interests. This phenomenon shows that it is indispensable for preserving the relation-specific information. As a result, a better solution is eagerly needed to explore the heterogeneous edges for network embedding. 

Over the past two decades, multi-view learning has achieved promising performance on mining heterogeneous data, such as clustering~\cite{liu2013multi}, natural language processing~\cite{phillips2002exploiting} and image processing~\cite{li2011difficulty}. In view of this, in this paper, we formalize the problem as multi-view learning from multiple network views. Specifically, given a heterogeneous network, as we focus on exploring the heterogeneous edges, we assume its nodes are homogeneous but the edges are heterogeneous. The heterogeneous edges indicate that there are many types of proximity between nodes in such a network. Each proximity can be defined as a view of the network, thus multiple proximities further generate a multi-view network. Then we solve a multi-view network embedding problem to learn a better network representation for such networks.

\begin{figure*}[t]
    \centering
    \includegraphics[scale=0.4]{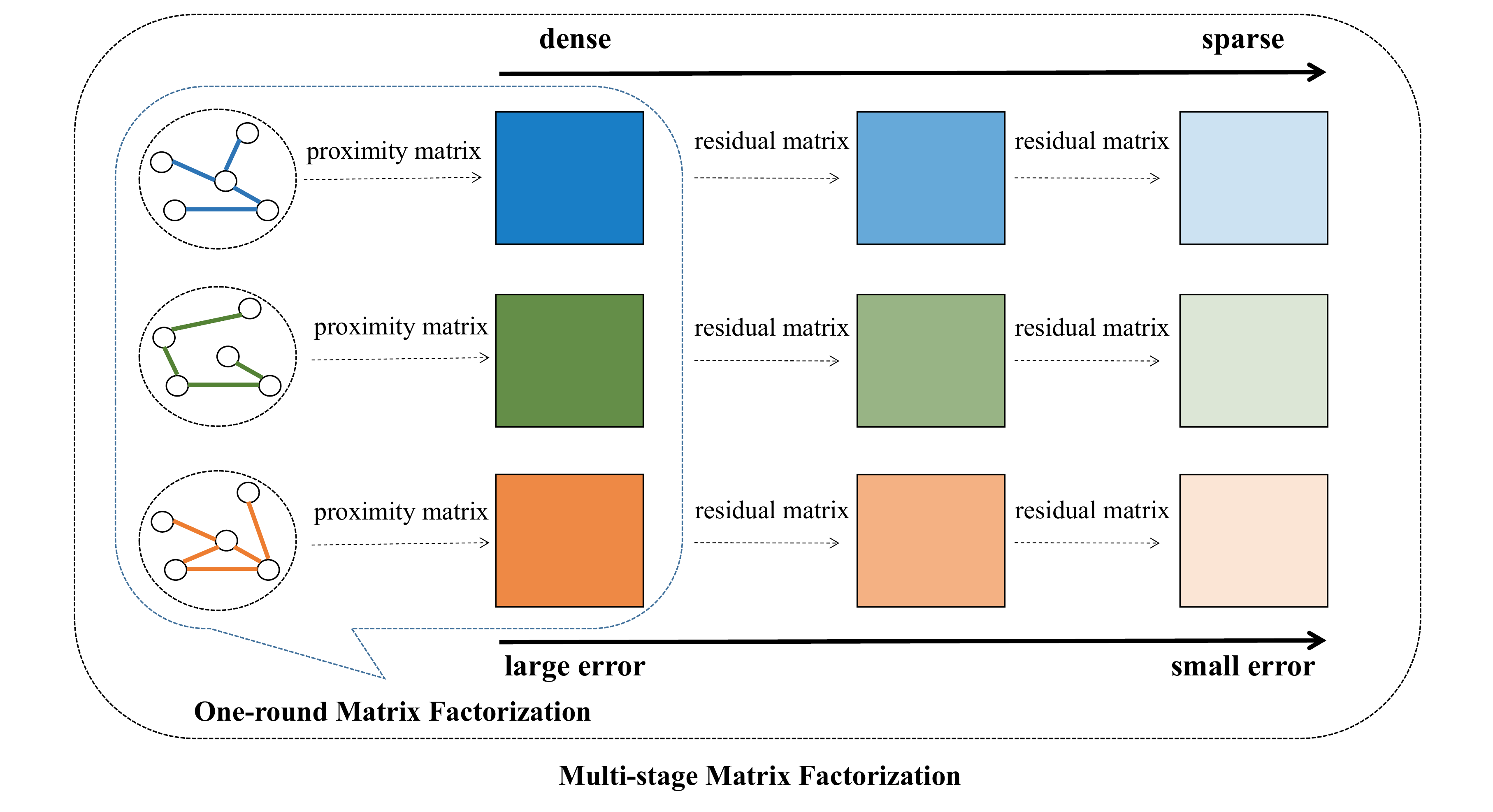}
    \vspace{-3mm}
    \caption{Illustration of the multi-stage matrix factorization. Given a multi-view network with $3$ views which correspond to three types of interpersonal relations. The difference between one-round matrix factorization and multi-stage matrix factorization lies in the following points. (1). The one-round matrix factorization process decomposes dense matrices only once resulting a large approximation error, which usually doesn't satisfy the low rank condition. (2). The multi-stage matrix factorization process successively decomposes the residual matrices in a forward manner to reduce the approximation error. With the stage going deep, the matrices will be more and more sparser, thus satisfy the low rank condition gradually.}
    \label{multi_view_network}
\end{figure*}

There are various network embedding methods~\cite{perozzi2014deepwalk,tang2015line,grover2016node2vec}, and it has been proven that many of them are equivalent to matrix factorization essentially~\cite{qiu2018network}. In consideration of the extensive universality of matrix factorization in multi-view learning~\cite{liu2013multi,ou2016multi,li2018multi}, it is more suitable for solving the multi-view network embedding problems. Although latest matrix factorization works \cite{zhang2018arbitrary,wang2017community} have successfully transformed single-view networks into representations, there still exist several limitations when matrix factorization is applied to represent multi-view networks. Firstly, traditional matrix factorization methods usually cause a relatively large error when it is applied to network embedding. In most cases, the proximity matrix to be decomposed is usually a dense matrix, because the proximity matrix not only reflects the first-order proximity information like the adjacency matrix, but also considers extra high-order proximity information of networks. Since a one-round matrix factorization process just decomposes such proximity matrix once, without further considering the residual part, it usually leads to a relatively large error. As a result, a one-round matrix factorization may obtain a sub-optimal representation, which is insufficient to encode the structural information in the original network.

Secondly, network data are usually partial distinct as well as partial correlated between different views. This fact results in two categories of information, i.e. consensus information and unique information. In other words, different views may carry consensus information, at the same time, a specific view may carry some unique information that the others don’t have. Current multi-view matrix factorization methods usually preserve either consensus information~\cite{liu2013multi} or unique information~\cite{christoudias2012multi} in multiple views. It is unquestionable that consensus information and unique information both play important roles in multi-view networks, thus it is necessary to consider both of them simultaneously for representing such networks, which will be illustrated in Sec.~\ref{consensus_vs_unique}. 

As shown in Fig.~\ref{multi_view_network}, the one-round matrix factorization usually decomposes dense matrices with a large approximation error. To overcome the above limitations, we design a \textbf{M}ulti-stage \textbf{N}on-negative \textbf{M}atrix \textbf{F}actorization model, named as \textbf{MNMF}, which factorizes the residual matrices iteratively. The residual matrices denote the approximation error between the factored matrices and the recovered matrices. Intuitively, the approximation error can be reduced by decomposing residual matrices in a forward manner. With the stage going deeper, the residual matrices will become sparser, meanwhile the approximation error will reduce gradually. As the multi-stage structure of MNMF model is based on the idea of gradient boosting, it gives interpretability for the advantages of the MNMF model over one-round NMF, which will be discussed in Sec.~\ref{discuss}.

Since the approximation error is reduced, there is no doubt that our model can better preserve the information of networks. Besides, the multi-stage structure of our model makes it easy to capture the consensus and unique information of networks concurrently and explicitly. To be specific, we run united NMF stage-by-stage on residual matrices in the first several stages. This process aims to model the consensus information between different views. After that, we further allow each view to preserve its unique information, thus the residual matrices will be decomposed independently in the last several stages. Finally, we aggregate the embedding results of both the united NMF stages and the independent NMF stages to obtain the final network representation.  

Overall, the main contributions are summarized as follows:
\squishlist

\item We propose a multi-stage NMF model, namely MNMF, to  learn a representation for a network from its multiple views. The multi-stage structure not only reduces the approximation error for matrix factorization, but also gives the feasibility to explicitly preserve the consensus and unique information of multiple views in different stages.
\item By performing an extensive analysis, we explain the differences and advantages between the MNMF model and the one-round non-negative matrix factorization. We show that the MNMF model is a gradient boosting-like approach, which can successively decompose the residual matrices to yield better network representations.

\item Experimental results on node classification and node clustering tasks demonstrate that the MNMF model is superior to the existing dominant baseline approaches. In addition, ancillary experiments also help to explain the advantages of our model.
\squishend

The rest of this paper is organized as follows. Sec. ~\ref{related} reviews the related work. Sec. ~\ref{method} formalizes the problems and introduces the model. Sec. ~\ref{exp} describes the design of experiments and reports the experimental results. Sec. ~\ref{conclusion} presents the conclusions and future works.

\section{Related Work}\label{related}

\subsection{Network Embedding}
Network embedding aims to find a nonlinear function to embed the raw network data into a low-dimensional latent space. Since the Word2vec \cite{mikolov2013distributed} was invented, the idea of skip-gram model has been borrowed to various network embedding methods, such as Deepwalk \cite{perozzi2014deepwalk}, LINE \cite{tang2015line} and node2vec \cite{grover2016node2vec}. Besides, various network embedding methods based on matrix factorization have been proposed, such as GraRep \cite{cao2015grarep} and HOPE \cite{ou2016asymmetric}. Recently, a study \cite{qiu2018network} proves that many network embedding methods inspired by skip-gram model are equivalent to matrix factorization and proposes the NetMF model for network embedding. In addition, heterogeneous network embedding methods have an increasing interests from industry and academia. For example, metapath2vec \cite{dong2017metapath2vec}, HIN2vec \cite{fu2017hin2vec} and RHINE~\cite{lu2019relation} utilize meta path to encode the structural and semantic information into network embeddings. PTE \cite{tang2015pte} is an extension of LINE applied in heterogeneous networks. EOE~\cite{xu2017embedding} is designed for embedding the distinct but related homogeneous networks via measuring the similarity between nodes in different networks. DMNE~\cite{ni2018co} mainly focus on modeling the cross-network relationships between multiple networks, leading to further refined network embeddings. Another rising line of network embedding works is graph neural networks \cite{kipf2016semi,hamilton2017inductive,morris2019weisfeiler}. However, most of them are unable to be applied to multi-view networks directly, and need labeled information for training. Compared with that, our proposed method is totally unsupervised and the learned representations can be general to downstream tasks. 

Apart from above methods, learning representations for multi-view networks have not been well explored. MVE~\cite{qu2017attention} first learns single view representations and then combines them via weighted averaging. MNE~\cite{zhang2018MNE} learns multiple relationships via a joint training strategy. However, both of them fail to consider the complex relationships between network views.  Although there are some works~\cite{lu2019auto,xu2019multi} preserving consensus and unique information simultaneously, they only consider low-order proximity between nodes, which makes them hard to achieve satisfactory performance. Compared with them, our model explicitly preserves consensus and unique information in high order proximity matrices via two different factorization manners. 
Moreover, since tensor factorization~\cite{de2009survey,lu2011survey} is an extension of matrix factorization, it is naturally applicable to learning the embedding of heterogeneous networks. Generally, the unified proximity matrix of a graph with heterogeneous edges can be treated as a collection of proximity matrices from different views, thus a great deal of efforts~\cite{fernandes2018dynamic,fu2015joint,kolda2008scalable, balazevic2019tucker} have been devoted to modeling heterogeneous edges by tensor factorization.

\subsection{Multi-View Learning}
Another related work is about multi-view learning. Some traditional multi-view learning algorithms, such as co-training \cite{kumar2011co}, co-clustering \cite{zhao2017multi}, cross-domain fusion \cite{elkahky2015multi} and sub-space learning \cite{wang2015deep}, have been widely used in network analysis. Considering multiple views will be helpful in various network mining tasks such as community detection~\cite{brodka2011degree} and link prediction~\cite{matsuno2018mell}. Non-negative matrix factorization (NMF) is also widely used for multi-view learning due to its interpretability. Currently, most of NMF based network embedding works are designed for single-view networks, such as BoostNE \cite{li2018multi} and M-NMF\cite{wang2017community}. Besides, existing multi-view NMF based methods \cite{ou2016multi,liu2013multi,xing2019multi} usually fail to obtain a satisfactory result for representing multi-view networks. Most of them jointly decompose matrices of multiple views in one-round, which can not preserve complete information of multi-view networks as well as causing large approximation error. In this work, the proposed algorithm is a general solution to mine multi-view networks since it can learn low-dimensional representations for social computing tasks. Although there exist some methods considering simultaneous or joint factorization manner to decompose multiple matrices together such as CMF~\cite{singh2008relational}, Pseudo-deflation based NMF~\cite{kim2015simultaneous} and GeoMF~\cite{lian2014geomf}, these methods are similar to ours but still have significant differences from ours. Moreover, our method allows an adjustable balance between joint factorization and independent factorization through controlling the number of stages for each manner, leading to better effectiveness 
and flexibility of our method.

\section{The Proposed MNMF Model}\label{method}

\subsection{Preliminaries}
We first give the definition of a multi-view network:
\begin{definition}\label{def1}
{\textbf{Multi-View Network}}
A multi-view network is defined as $\mathcal{G} = \left\{\mathcal{V},\mathcal{G}_1,\mathcal{G}_2,\cdots,\mathcal{G}_K\right\}$, where $\mathcal{V}$ is the node set shared by all views, $K$ is the number of all observed views, and $\mathcal{G}_k~(1\leq k \leq K)$ is the $k$-th view. $\mathcal{G}_k$ can be regarded as a single-view network, reflecting a single and distinct relationship among nodes.
\end{definition}

Previous matrix factorization based works have been successfully applied to learn representations for networks, which can be divided into two categories. One is starting from the network adjacency matrix $\bm{A}$ and decomposing a high order transition matrix, such as GraRep~\cite{cao2015grarep}. The other is starting from spectral graph theory and decomposing the equivalent matrix for existing skip-gram model based network embedding methods ~\cite{qiu2018network}. Along with the latter, we study Deepwalk \cite{perozzi2014deepwalk} with $b$ negative samples and $T$-sized window as an example since it is more general and efficient than other methods, which is equivalent to decompose the following matrix of a network $\mathcal{G}$:
\begin{equation}
    \bm{X} = \frac{vol(\mathcal{G})}{bT}\left(\sum_{t=1}^T \bm{P}^t\right)\bm{D}^{-1}
    \label{dw_matrix}
\end{equation}
where $\bm{D}$ is the degree matrix for network $\mathcal{G}$, $vol(\mathcal{G}) = \sum_{i}\sum_{j}\bm{A}_{i,j}$,  $\bm{A}_{i,j}$ denotes the edge weight between vertices $i$ and $j$, and $\small{\bm{P} = \bm{D}^{-1}\bm{A}}$. 

We choose non-negative matrix factorization~(NMF)~\cite{lee1999learning} to learn network representations, as NMF enforces the inputs to be expressed as linear additive combinations, thus has the property of interpretability. Given a network $\mathcal{G}$, Eq.~\eqref{dw_matrix} cannot be guaranteed to be a non-negative matrix. As a result, we reformulate it as $\footnotesize{\bm{M}_{i,j} = log(\max (\bm{X}_{i,j},1))}$ to ensure $\bm{M}$ is a non-negative matrix. Original matrix $\bm{X}$ is a dense matrix, despite $\bm{M}$ would be sparser it is still a dense matrix. In order to obtain the representation, standard NMF only performs the following one-round matrix factorization process:
\begin{equation}
    \centering
    \text{min}\quad \Vert \bm{M} - \bm{U}\bm{V} \Vert_F^2, \quad s.t. \quad \bm{U} \geq 0, \bm{V} \geq 0
    \label{nmf}
\end{equation}
where $\bm{U}$ can be interpreted as the representations of nodes that act as “center” nodes while $\bm{V}$ as “context” nodes~\cite{levy2014neural} in Deepwalk. However, Eq.~\eqref{nmf} usually causes a relative large approximation error, as $\bm{M}$ is a dense matrix that makes the low-rank condition untenable. Taking approximation error into consideration, Eq.~\eqref{nmf} can be formulated as:
\begin{equation}
    \centering
    \bm{M} \approx \widetilde{\bm{M}} = \bm{U}\bm{V} \Rightarrow \bm{M} = \widetilde{\bm{M}} + \bm{\Delta}
\end{equation}
where $\bm{\Delta}$ introduces the approximation error. The smaller $\bm{\Delta}$ means the smaller approximation error and the better representation. Obviously, we can further decompose $\bf{\Delta}$ in a forward manner to reduce the approximation error, thus a multi-stage NMF is introduced as following:
\begin{equation}
\begin{aligned}
    \text{min} &\quad \Vert \underbrace{\overbrace{\bm{M} - \bm{U}_1 \bm{V}_1}^{\bm{\Delta}_1} - \bm{U}_2\bm{V}_2}_{\bm{\Delta}_2} - \cdots - \bm{U}_q\bm{V}_q \Vert_F^2 \\
\Rightarrow    \text{min} &\quad \Vert \bm{M} - \sum_{l=1}^q \bm{U}_l \bm{V}_l \Vert_F^2, \quad s.t. \quad \bm{U}_l \geq 0, \bm{V}_l \geq 0
\end{aligned}
\label{boostnmf}
\end{equation}

In Eq.~\eqref{boostnmf}, the NMF process is performed by $q$ stages, thus the approximation results can be learned by $q$ individual learners. Similar to ensemble learning, it usually achieves better performance and lower approximation error. However, when coming across multi-view networks, we still need to handle two problems. One is how to adapt Eq.~\eqref{boostnmf} to learn representations for multi-view networks, since there are multiple views to be considered. The other is how to preserve consensus and unique information simultaneously, since they are essential to represent the network comprehensively and accurately. Consequently, we formally define the problem of multi-view embedding as following:
\begin{problem}\label{prob1}
{\textbf{Multi-View Network Embedding}} Given a multi-view network $\mathcal{G} = \{\mathcal{V},\mathcal{G}_1,\mathcal{G}_2,\cdots,$ $\mathcal{G}_K \}$, the number of stage $L$ and the embedding dimension $d_l$ in $l$-th stage. The multi-stage multi-view network embedding problem is defined as finding a series of representations  $\bm{U}^l$ $\in$ \bm{$\mathbb{R}^{|\mathcal{V}|*d_l}$} ($d_l$ $\ll$ $|V|$, 1 $\le$ l $\le$ L) for nodes in the network so the consensus and unique information in multiple views can be gradually preserved when the stage goes from 1 to $L$.
\end{problem}

\subsection{The MNMF Model}

In order to solve the Problem~\ref{prob1}, we propose a simple yet effective model named as MNMF, which is illustrated in Fig.~\ref{MNMF}. Overall speaking, we design two manners for the multi-stage matrix factorization process: the united NMF for preserving the consensus information between different views, and the independent NMF for preserving the unique information of each view. Here we first preserve the consensus information and then deal with unique information.

\begin{figure}[t]
    \centering
    \includegraphics[scale=0.4]{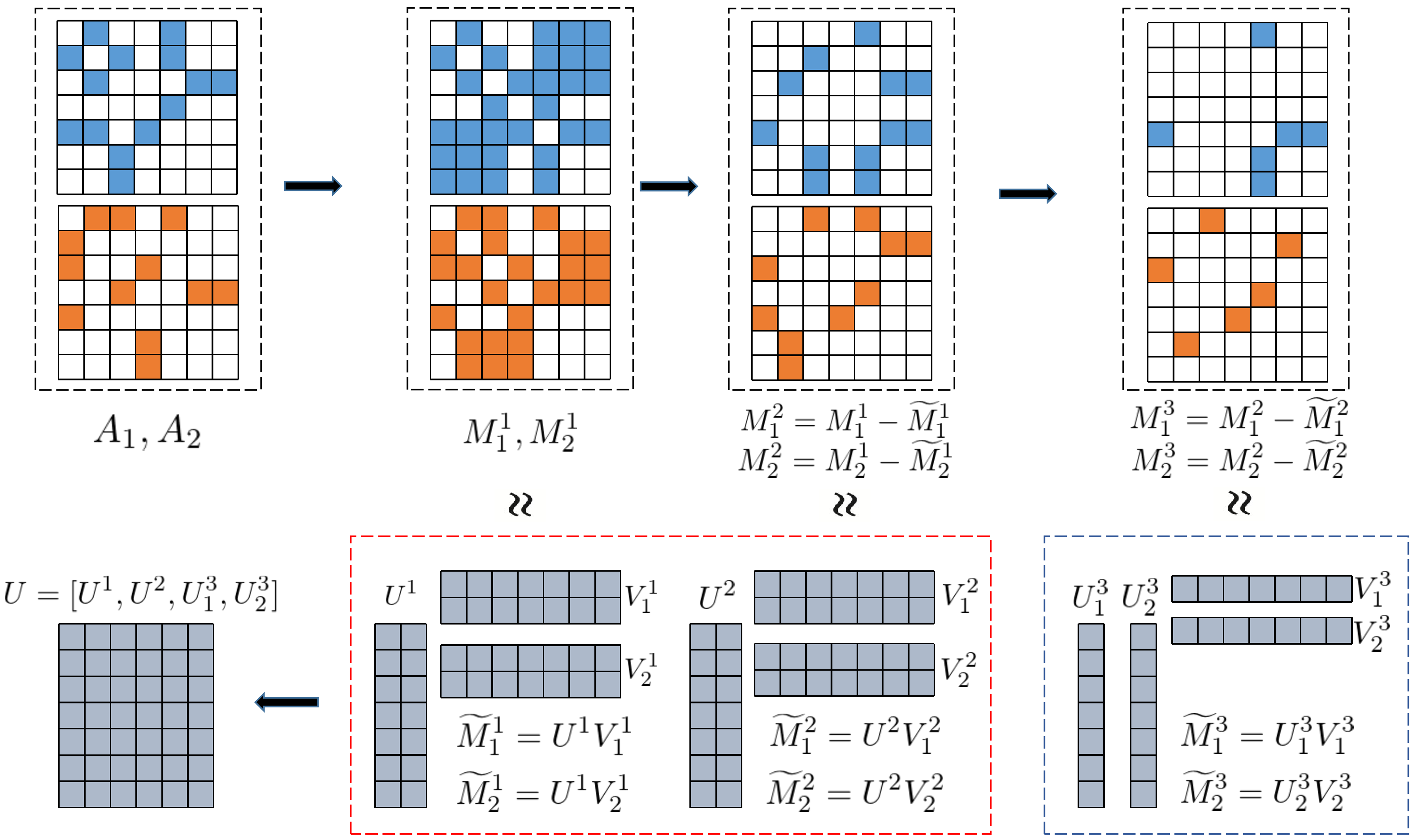}
    \caption{The illustration of MNMF model. \small Take two views and 3 stages as an example. By factorizing the not well approximated part $R_k^l$ iteratively, the MNMF model preserves consensus information (1-st and 2-nd stage) as well as unique information (3-rd stage) of different granularity. The robust representation of this multi-view network is obtained by assembling all representations from 3 stages. }
    \label{MNMF}
\end{figure}

In detail, at first, in order to generate consensus representations in the first $q$ stages, we design the following objective function:
\begin{equation}
    \centering
    \begin{aligned}
     \text{min} &\qquad \sum_{k=1}^K (\alpha_k)^{\gamma} \left \Vert \bm{M}_k - \sum_{l=1}^q \bm{U}^l \bm{V}_k^l \right \Vert_F^2 \\
     \text{s.t.} &\qquad \bm{U}^l\geq0, \bm{V}_k^l\geq 0,  \sum_{k=1}^K \alpha_k=1, \alpha_k \geq 0 \\
    \end{aligned}
    \label{naive_consensus}
\end{equation}
where $\alpha_k$ is weight coefficient for view $k$ and hyper-parameter $\gamma$ controls the weight distribution. Eq.~\eqref{naive_consensus} is inspired by preserving consensus information, and the $\bm{U}^l$ $\in$ $\mathbb{R}^{|\mathcal{V}|*d_l}$ represents the consensus information in stage $l$ since it is shared by all views. However, if we optimize different $\bm{U}^l$, $\bm{V}_k^l$ and $\alpha_k$ in Eq.~\eqref{naive_consensus} simultaneously  using exitsing methods, such as coordinate descent or block-coordinate descent \cite{Tseng2001Convergence}, it would break down the forward stage manner since variables can be updated in any order. One possible solution is successively decomposing the previous not approximated parts (i.e. $\bf{\Delta}$) in multiple stages. The not approximated parts are defined as residual matrices:
\begin{equation}
    \bm{R}_k^l = \left\{
    \begin{aligned}
    & \log\left({\max\left(\frac{vol(\mathcal{G}_k)}{bT}(\sum_{t=1}^{T} {\bm{P}_k}^t) \bm{D}_k^{-1},1\right)}\right), l = 1\\
    & \max(\bm{R}_k^{l-1} - \bm{U}^{l-1}\bm{V}_k^{l-1}, 0), 2 \le l \le q
    \end{aligned}
\right.
\label{residual}
\end{equation}
Note here the $\max$ operation is adopted to clip negative entries as zero. And it is easy to find that $\bm{R}_k^1 = \bm{M}_k$. Then we solve a close-form problem of Eq.~\eqref{naive_consensus} at each stage in a forward manner, which results in the reformulated objective function in $l$-th ($1 \le l \le q$) stage :
\begin{equation}
    \begin{aligned}
    \text{min} &\qquad \sum_{k=1}^K (\alpha_k^l)^{\gamma}{\Vert \bm{R}_k^l - \bm{U}^l\bm{V}_k^l \Vert_F^2}, \\ 
    \text{s.t.} &\qquad \bm{U}^l\geq0, \bm{V}_k^l\geq 0,  \sum_{k=1}^K \alpha_k^l=1, \alpha_k^l \geq 0 
    \end{aligned}
    \label{consensus}
\end{equation}
where $\alpha_k^l$ is weight coefficient for view $k$ in $l$-th stage and $\gamma$ controls the weight distribution. Eq.~\eqref{consensus} is repeated stage-by-stage, with the stage going deeper, the residual matrices will become sparser and easy to be factorized, at the same time the approximation error will reduce gradually. At the last $L-q$ stages, we tend to preserve unique information by NMF in each view at $l$-th ($q < l \le L$) stage:
\begin{equation}
    \text{min} \Vert \bm{R}_k^l - \bm{U}_k^l \bm{V}_k^l \Vert_F^2, \bm{U}_k^l\geq0, \bm{V}_k^l\geq0, k=1,2,\cdots,K
    \label{complement}
\end{equation}
where $\bm{U}_k^l$ represents the unique information of view $k$ when independently generated, the residual matrices are redefined as
\begin{equation}
    \bm{R}_k^l = \max(\bm{R}_k^{l-1} - \bm{U}_k^{l-1} \bm{V}_k^{l-1}, 0)
    \label{residual_unique}
\end{equation}
Note that Eq.~\eqref{complement} is also repeated for all views stage-by-stage. 

The final network representation will be the ensemble of the results in all stages obtained by an aggregator, which can be sum, pooling or other designed functions. For simplicity, in this paper, we choose to concatenate the generated embeddings of all stages to preserve both consensus and unique information simultaneously. Moreover, our model can also be applied to such multi-view networks that have a few nodes unobserved in each view. In this case, we just need to pad those unobserved nodes' entries as zero in the adjacency matrix of each view.

\subsection{Discussion} \label{discuss}
\noindent (1) \textbf{United NMF}. The MNMF model contains two matrix factorization manners, i.e. united non-negative matrix factorization and independent non-negative matrix factorization, corresponding to Eq.~\eqref{consensus} and Eq.~\eqref{complement} respectively. Here we mainly discuss united non-negative matrix factorization process since Eq.~\eqref{complement} is a naive NMF  process that has been well explored in literatures \cite{lee2001algorithms,cai2008non}. 

The idea behind Eq.~\eqref{consensus} is utilizing a shared basis $\bm{U}^l$ to encode consensus information among multiple views. This strategy can be understood as that it essentially requires nodes in each view to share a basis vector (each row in $\bm{U}^l$ corresponds to a node) as their representation during matrix factorization, but the coordinates in each view (each column in $\bm{V}_k^l$) under the basis are different. As mentioned before, $\bm{U}$ can be interpreted as the representations of nodes that act as “center” nodes, while $\bm{V}$ as “context” nodes in skip-gram model. Similar to the idea in \cite{qu2017attention}, the united matrix factorization process allows the “center” representations to be shared across multiple views, The main idea is similar to MVE but the differences lie in two points: (1) The MVE utilizes the shared node embedding to regularize the view-specific node embeddings then fuses them to yield the consensus representation, i.e. the final network embedding; compared with that, MNMF utilizes view-specific node embeddings to regularize the shared node embedding thus it can be treated as the consensus representation naturally, without the additional fusing step. Moreover, MNMF designs the independent factorization manner to consider the view-specific information explicitly. (2) The fusing step in MVE needs extra labels served as supervised signals while the MNMF is totally unsupervised, leading to better generalization of the learned embedding, which brings two benefits:
\squishlist
\item Different views will be projected into the same semantic space via $\bm{V}_k^l$ as a projection matrix for transforming them back, thus it is natural to regard $\bm{U}^l$ as their consensus information; 
\item It allows multiple views to collaborate together to learn their consensus information, thus building a bridge between different views to make them influence mutually and promote implicitly.
\squishend

\noindent (2) \textbf{Multi-stage NMF vs One-round NMF}. Here we explain the difference between the multi-stage NMF process and the one-round NMF process. It is simple to prove that the original objective function in Eq.~\eqref{naive_consensus} is equivalent to a one-round NMF process as but essentially different, we will explain the differences in detail in this part:
\begin{equation}
    \centering
    \begin{aligned}
     \text{min} &\qquad \sum_{k=1}^K (\alpha_k)^{\gamma} \left \Vert \bm{M}_k - \sum_{l=1}^q \bm{U}^l \bm{V}_k^l \right \Vert_F^2 \\
     =  \text{min} &\qquad \sum_{k=1}^K (\alpha_k)^{\gamma} \left \Vert \bm{M}_k - \bm{U}\bm{V}_k  \right \Vert_F^2 \\
     \text{s.t.} &\qquad \bm{U}^l\geq0, \bm{V}_k^l\geq 0,  \sum_{k=1}^K \alpha_k=1, \alpha_k \geq 0 \\
    \end{aligned}
\end{equation}
where $\bm{U} = [\bm{U}^1 \quad \bm{U}^2 \quad \cdots \quad \bm{U}^q]$ and $\bm{V}_k = [\bm{V}_k^1 \quad \bm{V}_k^2 \quad \cdots \quad \bm{V}_k^q]^\mathrm{T}$. However, the main difference between Eq.~\eqref{naive_consensus} and our proposed multi-stage NMF process (illustrated as Eq.~\eqref{residual} and Eq.~\eqref{consensus}) lies in the optimization method. Eq.~\eqref{naive_consensus} optimizes all variables simultaneously while Eq.~\eqref{consensus} is solved in a successive way, which means once the solution $\bm{U}^l$, $\bm{V}_k^l$ and $\alpha_k^l$ are obtained at a certain stage $l$, they will be fixed for the remaining stages. This optimization strategy is similar to a functional gradient boosting~\cite{friedman2002stochastic} method. To this end, for a certain view $k$, we ignore the view weight coefficient and let $f_k^{(q)}$, $L$ be
\begin{equation}
    \begin{aligned}
    f_k^{(q)} = f(\bm{U}^1,\cdots,\bm{U}^q,\bm{V}_k^1,\cdots,\bm{V}_k^q) = \sum_{l=1}^q \bm{U}^l \bm{V}_k^l,\\
    L(\bm{M}_k,f_k^{(q)}) = \left \Vert \bm{M}_k - f_k^{(q)}   \right \Vert = \left \Vert \bm{M}_k - \sum_{l=1}^q \bm{U}^l \bm{V}_k^l  \right \Vert_F^2
    \end{aligned}
    \label{gbnmf}
\end{equation}
Here we can easily compute the gradient of Eq.~\eqref{gbnmf}, which is expressed as
\begin{equation}
\begin{aligned}
    g_q &= \left\lbrack \frac{\delta L (\bm{M}_k,f_k^{(q)})}{\delta f_k^{(q)}} \right\rbrack_{f_k^{(q)} = f_k^{(q-1)}} \\
    &= 2(\bm{M}_k - f_k^{(q-1)}) = 2(\bm{M}_k - \sum_{l=1}^{q-1} \bm{U}^l \bm{V}_k^l)
\end{aligned}
\label{grad_nmf}
\end{equation}
Based on Eq.~\eqref{grad_nmf}, by imposing the non-negative constraints to $f_k^{(q)}$ because of $\bm{U}^l $,$\bm{V}_k^l$ $\geq 0$ and ignoring the constant, the projected gradient can be derived by setting $l=1,\cdots,q$
\begin{equation}
\begin{aligned}
    \bm{M}_k &= \log\left({\max\left(\frac{vol(\mathcal{G}_k)}{bT}(\sum_{t=1}^{T} {\bm{P}_k}^t) \bm{D}_k^{-1},1\right)}\right) \\
    \bm{R}_k^l &=  \left \lbrack \left \lbrack \left\lbrack \bm{M}_k - \bm{U}^1\bm{V}_k^1 \right\rbrack_+ - \bm{U}^2\bm{V}_k^2 \right\rbrack_+  \cdots - \bm{U}^{l-1}\bm{V}_k^{l-1} \right\rbrack_+
\end{aligned}
\label{resudial_consensus_v2}
\end{equation}
where $[x]_+$ is applied for the non-negative constraint. Obviously, the Eq.~\eqref{resudial_consensus_v2} is same as Eq.~\eqref{residual}, thus each stage of the Eq.~\eqref{consensus} is established on the gradient descent direction of the previous stage's loss function. In a same vein, we can also prove that the objective function in Eq.~\eqref{complement} is a closed-form gradient boosting method. As a result, we conclude that the MNMF model is close to a gradient boosting approach, which illustrates the reason for the superiority of the MNMF model to existing NMF models.

\subsection{Optimization}
Since the objective function in Eq.~\eqref{consensus} is not convex over all variables simultaneously, we optimize them by successively updating only one variable at a time, while fixing the other variables. Here we directly drive the update rules for $\bm{U}^l$ and $\bm{V}_k^l$ as\footnote{The detailed derivation process is available in Appendix.}:
\begin{equation}
    \begin{aligned}
        \bm{U}^l &\leftarrow \bm{U}^l \odot \frac{\sum_k^K (\alpha_k^l)^{\gamma} \bm{R}_k^l {\bm{V}_k^l}^\mathrm{T}}{\sum_k^K (\alpha_k^l)^{\gamma} \bm{U}^l \bm{V}_k^l {\bm{V}_k^l}^\mathrm{T}} \\
        \bm{V}_k^l &\leftarrow \bm{V}_k^l \odot \frac{{\bm{U}^l}^\mathrm{T} \bm{R}_k^l}{{\bm{U}^l}^\mathrm{T} \bm{U}^l \bm{V}_k^l}
    \end{aligned}
    \label{update4consensus}
\end{equation}
where $\odot$ and $\frac{[\cdot]}{[\cdot]}$ are element-wise multiplication and division respectively.

Let's denote $\bm{W}_k^l = {\Vert \bm{R}_k^l - \bm{U}^l\bm{V}_k^l \Vert_F^2}$, then the update rule for $\alpha_k^l$ is in Eq.~\eqref{alpha_update}:
\begin{equation}
    \centering
    \alpha_k^l \leftarrow \frac{\left(\gamma \bm{W}_k^l\right)^{\frac{1}{1-\gamma}}}{\sum_{k=1}^K\left(\gamma \bm{W}_k^l\right)^{\frac{1}{1-\gamma}}}
\label{alpha_update}
\end{equation}

For Eq.~\eqref{complement}, we use the same optimization techniques as described in~\cite{lee2001algorithms}. 

In summary, the pseudo code of the MNMF model is shown in Algorithm.~\ref{MNMF_pcode}.

\begin{algorithm}[htbp]
\LinesNumbered
\KwIn{ A network with $K$ views, denoted as $\mathcal{G} = \left\{\mathcal{V},\mathcal{G}_1,\mathcal{G}_2,\cdots,\mathcal{G}_K\right\}$\;
The hyper-parameter $\gamma$, the number of consensus stages $q$, the number of total stages $L$, window-size $T$, the number of negative samples $b$ and embedding size $d$ in each stage.}
\KwOut{The representation $\bm{Z}$ for the network $\mathcal{G}$.}
$\triangleright$ \textbf{Consensus NMF process:} \\
\For{l = 1, 2, $\cdots$ ,q}{
According to Eq.~\eqref{residual}, calculate residual matrix for each view;\\
Initialize all variables with a uniform distribution $\mathcal{U}$(0, 0.02);\\
\Repeat
{\text{Convergence}}
{
According to Eq.~\eqref{update4consensus}, fixing other variables, update $\bm{U}^l$;\\
According to Eq.~\eqref{update4consensus}, fixing other variables, update $\bm{V}_k^l$ for each view;\\
According to Eq.~\eqref{alpha_update}, fixing other variables, update ${\alpha}_k^l$ for each view;\\
}
Obtain consensus representation $\bm{U}^l$;\\
}
$\triangleright$ \textbf{Unique NMF process:} \\
\For{l = q+1, q+2, $\cdots$, L}{
According to Eq.~\eqref{residual_unique} or Eq.~\eqref{residual} (only for $l=q+1$), calculate residual matrix for each view\;
\For{k=1,2,$\cdots$,K}{
According to Eq.~\eqref{complement}, perform a standard NMF;\\
Obtain unique representation $\bm{U}_k^l$;\\
}
}
$\triangleright$ \textbf{Aggregator process:} \\
$\bm{Z}$ = $\bm{concat}$[consensus representations, unique representations];\\
\Return Network representation $\bm{Z}$
\caption{The pseudo-code of MNMF Model}
\label{MNMF_pcode}
\end{algorithm}

\subsection{Convergence Analysis}
The convergence of the update rules can be analyzed as follows: We can divide the Eq.\eqref{consensus} into three subproblems: (1) The first subproblem is fixing other variables and updating $\bm{U}^l$; (2) The second subproblem is fixing other variables and updating $\bm{V}_k^l$; (3) The third subproblem is fixing other variables and updating $\alpha_k^l$. Each subproblem is a convex problem with respect to one variable. As a result, by solving the subproblems alternatively, our proposed algorithm will guarantee that we can find the optimal solution to each subproblem and finally, the Eq.~\eqref{consensus} will converge to a local minimum. Similarly, the Eq.~\eqref{complement} is also guaranteed to converge to a local minimum.

\subsection{Complexity Analysis} \label{complexity}
Briefly speaking, the time complexity is positively related to the number of nonzero entries of $\bm{R}_k^l$, denoted as $nnz(\bm{R}_k^l)$, thus the total number of nonzero entries of all $K$ views in $l$-th stage is denoted as ${nnz(\bm{R}^l) = nnz(\bm{R}_1^l) + nnz(\bm{R}_2^l) + \cdots + nnz(\bm{R}_K^l)}$. To simplify the analysis, we assume the dimensions in all layers are the same, denoted as $d$, number of stages in Eq.~\eqref{consensus}, Eq.~\eqref{complement} as $s_1$, $s_2$ respectively, and the number of iterations as $t$.  The total time complexity for Eq.~\eqref{consensus} is $\mathcal{O}(td \times (nnz(\bm{R}^1) + nnz(\bm{R}^2) + \cdots + nnz(\bm{R}^{s_1}))$ and for Eq.~\eqref{complement} it is $\mathcal{O}(td \times (nnz(\bm{R}^{s_1 + 1}) + nnz(\bm{R}^{s_1 + 2}) + \cdots + nnz(\bm{R}^{s_1 + s_2})))$. The superscript represents the stage id. As illustrated in Fig~\ref{MNMF}, due to the multi-stage structure of MNMF, the residual matrices to be decomposed become sparse quickly, so most of the time our algorithm decomposes a sparse matrix with an acceptable time complexity.

\section{Experiments}\label{exp}
In this part, we aim to answer three questions:
\squishlist
\item{\textbf{Q1:}} How does the proposed method compare with the state-of-the-art network representation baselines?
\item{\textbf{Q2:}} How do the different modules (i.e. multi-stage matrix factorization structure, view weight strategy and two factorization manners) influence the results?
\item{\textbf{Q3:}} How do the different hyper-parameters affect the performance?
\squishend

\begin{table*}[t]
\centering
\renewcommand\arraystretch{1.2}
\setlength{\tabcolsep}{1.5mm}
\topcaption{Overview of Datasets.}
\label{dataset}
\begin{tabular}{c|ccccc}
\hline
\textbf{Dataset} & \textbf{Views} & \textbf{Nodes} & \textbf{Edges} & \textbf{Labels} & \textbf{Configs}         \\ \hline\hline
AMiner\_small     & 3              & 8,438          & 2,433,356      & 8               & stages=16, dimensions=256, T=5, b=5 \\
AMiner\_large     & 3              & 37,555          & 2,885,964      & 8               & stages=8, dimensions=256, T=5, b=5  \\
PPI              & 6              & 4,328           & 1,661,756      & 50              & stages=8, dimensions=256, T=5, b=5  \\
Flickr           & 2              & 13,696          & 2,489,406      & 169             & stages=16, dimension=256, T=5, b=5  \\ \hline
\end{tabular}
\end{table*}

\subsection{Experimental Settings}
We perform experiments on three multi-view networks. The details are shown in Table~\ref{dataset}.
\squishlist
\item \textbf{AMiner}: The AMiner network~\cite{tang2008arnetminer} contains three views: author-citation, co-authorship and text similarity. The edges in text similarity view connect authors with their top 10 nearest neighbors, and the similarity is calculated based on the title and abstract by TF-IDF. We only preserve authors in eight research fields as~\cite{dong2017metapath2vec}. The authors is treated as nodes and research fields as node labels. Node that the dataset with subscript 'large' preserves all the nodes while the dataset with subscript 'small' only preserves nodes whose degree exceeds 2.
\item \textbf{PPI}: The PPI network~\cite{franceschini2012string} is a human protein to protein interaction network. Six views are constructed based on co-expression, co-occurrence, database, experiment, structural and protein-fused information. We treat gene groups as node labels.
\item \textbf{Flickr}: The network built from flickr dataset includes two views~\cite{tang2009relational}. Friendship view is the friendship network among the blog owners. Tag-proximity view is the network whose edges connect each node to its top 10 nearest neighbors, and the proximity is calculated based on the user's tags. The community memberships are used as node labels.
\squishend

We compare our method with \textbf{single-view based baselines}: 
\squishlist
\item \textbf{Deepwalk} \cite{perozzi2014deepwalk}: It is the most common baseline for network embedding based on random walk and skip-gram model. According to the recommendation in the original paper, the number and the length of random walks for each node are set as 80 and 40 respectively and the window-size is set as 10 for better performance.
\item \textbf{BoostNE} \cite{li2018multi}: It is an ensemble-based matrix factorization method for single-view network embedding. Compared with it, our method is extended to multi-view networks and has two decomposition strategies. The window-size and negative samples are both set as 5, which is same as our model. In addition, the number of stages is same as our proposed model to guarantee fairness.
\item \textbf{NetMF} \cite{qiu2018network}: According to the discussion in the original paper of BoostNE, NetMF is a special case of BoostNE when the stage is set as 1, so we use the same parameter settings as BoostNE.
\squishend

We compare our method with existing \textbf{heterogeneous network embedding baselines}:
\squishlist
\item \textbf{PTE} \cite{tang2015pte}: It is an extension method of LINE~\cite{tang2015line} for heterogeneous network embedding. Multi-view network is a type of heterogeneous network, thus we adopt PTE to generate its representation by joint training. The number of negative samples is set as 5.
\item \textbf{Metapath2Vec} \cite{dong2017metapath2vec}: It utilizes metapath-based random walk plus skip-gram model to capture the heterogeneous information in a network. We set the number and length of walks same as deepwalk, and the window size is also same as deepwalk. We perform experiment using one of all possible meta-paths at a time, and report the best result.
\squishend

We also compare our method with existing \textbf{multi-view based baselines}:
\squishlist
\item \textbf{NMF-con}: It applies NMF to decompose Eq.~\eqref{dw_matrix} to get a $d$ dimensional representation for each view then concatenates these representations from all $K$ views to generate a unified representation with $K*d$ dimensions. 
\item \textbf{MNE} \cite{zhang2018MNE}: It is a novel method for multiplex networks to transform all views into a unified representation via modeling multiple relationships jointly. We also follow the original paper to set the additional vectors' dimension as 10.
\item \textbf{MultiNMF} \cite{liu2013multi}: It is a multi-view learning algorithm based on co-clustering, which only preserves consensus information via joint NMF on Eq.~\eqref{dw_matrix} in all views. Multiple views are assigned with equivalent importance, which is same as the original paper but different from our proposed algorithm.
\item \textbf{TuckER~\cite{balazevic2019tucker}}: It is a powerful tensor factorization framework for representing relational data, we use the entity embeddings as the features for our experiment. The parameter settings are same as the original paper. 
\item \textbf{MTNE-C} \cite{xu2019multi}: It jointly learns the view-specific embedding and view-sharing embedding to preserve consensus and unique information in multi-view networks. We follow the default experimental settings in our experiments.  
\item \textbf{MVE} \cite{qu2017attention}: It applies deepwalk to get the representation for each view and then combines them by weights learned from attention mechanism. The setting of random walk and skip-gram model is same as Deepwalk, and we use the default hyper parameters in the original paper.
\squishend

For our proposed MNMF model, the parameter setting is listed in Table.~\ref{dataset}. Besides, we assign the total dimensions to each stage equally. Due to the fact that unique information is more than consensus information (see Sec.~\ref{consensus_vs_unique}), the first quarter of stages are responsible for preserving consensus information and the rest of stages aim to preserve unique information. The hyper-parameter $\gamma$ is searched from $\left\{0.01,0.05,0.1,0.5,5,10,50,100\right\}$ to obtain the best results.  We keep all methods with the same representation dimensions except NMF-con. For all methods, we use the learned embeddings as features to train linear classifiers for multiclass classification, and train one-vs-rest classifiers for multilabel classification. The ratio of training data varies from 30\% to 70\%. For node clustering task, the embeddings of nodes are treated as features and the number of clusters is same as the number of labels. To assign node into different clusters, we adopt KMeans \cite{macqueen1967some} to fit node embeddings. We report the best results among multiple views for single-view based baselines. All our experiments are conducted on a server with Intel(R) Xeon(R) CPU E5-2680, two Tesla P100 GPUs and 250Gb Memory. The experimental environment is Ubuntu 14.04.5 with CUDA 9.0. All methods are implemented with Tensorflow 1.12.0. To guarantee a fair comparison, we repeat each method 10 times and the average metrics are reported. The reference code is available at \url{https://github.com/yusonghust/MNMF/tree/master}.

\subsection{Performance Comparison (Q1)}

\begin{table}[t]
\topcaption{Multi-class classification results on AMiner\_small dataset w.r.t. Micro-f1(\%) and Macro-f1(\%).}
\centering
\setlength{\tabcolsep}{1.5mm}
\renewcommand\arraystretch{1.2}
\begin{tabular}{c|c|cc|cc|cc}
\hline
\multirow{2}{*}{\textbf{Category}}     & \multirow{2}{*}{\textbf{Methods}} & \multicolumn{2}{c|}{\textbf{0.3}}     & \multicolumn{2}{c|}{\textbf{0.5}}     & \multicolumn{2}{c}{\textbf{0.7}}     \\ \cline{3-8} 
                                      &                                          & \textbf{Micro}    & \textbf{Macro}    & \textbf{Micro}    & \textbf{Macro}    & \textbf{Micro}    & \textbf{Macro} \\ \hline \hline
\multirow{3}{*}{\textbf{Single-View}}  & Deepwalk                                 & 68.5              & 66.6              & 74.3              & 73.4              & 73.4              & 72.0              \\
                                      & NetMF                                    & 77.5              & 76.3              & 78.8              & 77.6              & 80.6              & 79.2              \\
                                      & BoostNE                                  & 78.3              & 77.5              & 79.3              & 78.8              & 81.8              & 81.2     \\ \hline
\multirow{2}{*}{\textbf{Heterogeneous}}& PTE                                      & 54.2              & 51.1              & 57.7              & 58.8              & 59.4              & 57.6              \\
                                      & Metapath2vec                             & 78.6              & 77.1              & 79.0              & 78.3              & 80.6              & 80.2        \\ \hline
\multirow{7}{*}{\textbf{Multi-View}}   & NMF-con                                  & 77.3              & 76.3              & 78.9              & 78.0              & 80.1              & 79.5              \\
                                      & MNE                                      & 73.0              & 71.9              & 78.7              & 77.2              & 78.8              & 77.2              \\
                                      & MultiNMF                                 & 48.5              & 41.4              & 48.6              & 42.2              & 49.5              & 41.9              \\
                                      & TuckER                                   & 79.6              & 77.8              & 80.2              & 78.8              & 80.9              & 80.1              \\
                                      & MTNE-C                                   & 57.2              & 53.9              & 58.6              & 55.2              & 59.5              & 55.4              \\
                                      & MVE                                      & 72.6              & 70.9              & 78.1              & 77.2              & 77.7              & 76.8              \\ \cline{2-8} 
                                      & MNMF                                     & \textbf{82.3}     & \textbf{80.8}     & \textbf{83.1}     & \textbf{82.0}     & \textbf{84.1}     & \textbf{82.9}      \\
                                      & (Gain)                                   & +2.7              & +3.0              & +2.9              & +3.2              & +2.3              & +1.7               \\ \hline
\end{tabular}
\label{aminer_cls}
\end{table}

\begin{table}[t]
\topcaption{Multi-class classification results on AMiner\_large dataset w.r.t. Micro-f1(\%) and Macro-f1(\%).}
\centering
\setlength{\tabcolsep}{1.5mm}
\renewcommand\arraystretch{1.2}
\begin{tabular}{c|c|cc|cc|cc}
\hline
\multirow{2}{*}{\textbf{Category}}     & \multirow{2}{*}{\textbf{Methods}} & \multicolumn{2}{c|}{\textbf{0.3}}     & \multicolumn{2}{c|}{\textbf{0.5}}     & \multicolumn{2}{c}{\textbf{0.7}}     \\ \cline{3-8} 
                                      &                                          & \textbf{Micro}    & \textbf{Macro}    & \textbf{Micro}    & \textbf{Macro}    & \textbf{Micro}    & \textbf{Macro} \\ \hline \hline
\multirow{3}{*}{\textbf{Single-View}}  & Deepwalk                                 & 62.0              & 61.1              & 62.5              & 61.8              & 63.3              & 62.7              \\
                                      & NetMF                                    & 73.3              & 70.9              & 74.2              & 72.0              & 74.3              & 72.3              \\
                                      & BoostNE                                  & 77.7              & 76.7              & 78.0              & 77.4              & 78.6              & 77.7              \\ \hline
\multirow{2}{*}{\textbf{Heterogeneous}}& PTE                                      & 40.1              & 38.0              & 41.4              & 39.2              & 42.1              & 39.4              \\
                                      & Metapath2vec                             & 74.1              & 71.7              & 76.4              & 73.8              & 77.2              & 75.3              \\ \hline
\multirow{7}{*}{\textbf{Multi-View}}   & NMF-con                                  & 72.9              & 70.7              & 75.3              & 73.4              & 76.2              & 74.6              \\
                                      & MNE                                      & 66.8              & 66.1              & 68.3              & 67.6              & 69.4              & 68.9              \\
                                      & MultiNMF                                 & 38.9              & 35.4              & 39.7              & 38.5              & 41.0              & 39.0              \\
                                      & TuckER                                   & 77.2              & 76.4              & 78.1              & 77.1              & 78.8              & 77.9              \\
                                      & MTNE-C                                   & 44.0              & 40.8              & 45.1              & 42.6              & 46.2              & 44.2              \\
                                      & MVE                                      & 65.9              & 65.3              & 66.6              & 66.1              & 67.7              & 67.6              \\ \cline{2-8} 
                                      & MNMF                                     & \textbf{79.2}     & \textbf{78.0}     & \textbf{79.4}     & \textbf{78.3}     & \textbf{79.9}     & \textbf{78.9}      \\
                                      & (Gain)                                   & +1.5              & +1.3              & +1.3              & +0.9              & +1.1              & +1.0               \\ \hline
\end{tabular}
\label{aminer_large_cls}
\end{table}

\begin{table}[t]
\topcaption{Multi-class classification results on PPI dataset w.r.t. Micro-f1(\%) and Macro-f1(\%).}
\centering
\setlength{\tabcolsep}{1.5mm}
\renewcommand\arraystretch{1.2}
\begin{tabular}{c|c|cc|cc|cc}
\hline
\multirow{2}{*}{\textbf{Category}}     & \multirow{2}{*}{\textbf{Methods}} & \multicolumn{2}{c|}{\textbf{0.3}} & \multicolumn{2}{c|}{\textbf{0.5}} & \multicolumn{2}{c}{\textbf{0.7}} \\ \cline{3-8} 
                                      &                                   & \textbf{Micro}  & \textbf{Macro}  & \textbf{Micro}  & \textbf{Macro}  & \textbf{Micro}  & \textbf{Macro}  \\ \hline \hline
\multirow{3}{*}{\textbf{Single-View}}  & Deepwalk                          & 10.1            & 7.4             & 11.1            & 8.0             & 13.8            & 9.9             \\
                                      & NetMF                             & 19.6            & 16.8            & 19.6            & 17.0            & 19.9            & 18.2            \\
                                      & BoostNE                           & 18.6            & 16.5            & 18.6            & 17.2            & 19.4            & 17.8             \\ \hline
\multirow{2}{*}{\textbf{Heterogeneous}}& PTE                               & 19.1            & 12.0            & 21.8            & 15.4            & 23.2            & 17.1            \\
                                      & Metapath2vec                      & 20.8            & 15.7            & 21.6            & 15.8            & 21.9            & 16.5              \\ \hline
\multirow{7}{*}{\textbf{Multi-View}}   & NMF-con                           & 21.0            & 16.7            & 22.3            & 17.2            & 22.7            & 17.9              \\
                                      & MNE                               & 15.1            & 13.2            & 16.6            & 13.1            & 18.1            & 15.6             \\
                                      & MultiNMF                          & 7.2             & 6.1             & 7.4             & 6.4             & 7.5             & 6.9             \\
                                      & TuckER                            & 19.7            & 16.8            & 22.3            & 17.1            & 23.5            & 18.0            \\ 
                                      & MTNE-C                            & 9.8             & 5.3             & 10.4            & 5.9             & 12.9            & 6.4             \\
                                      & MVE                               & 13.1            & 10.6            & 14.3            & 11.8            & 15.9            & 12.2            \\ \cline{2-8} 
                                      & MNMF                              & \textbf{22.2}   & \textbf{17.0}   & \textbf{25.1}   & \textbf{19.7}   & \textbf{27.6}   & \textbf{19.5}   \\
                                      & (Gain)                            & +1.2            & +0.2            & +2.8            & +2.5            & +4.1            & +1.3            \\ \hline
\end{tabular}
\label{ppi_cls}
\end{table}

\begin{table}[t]
\topcaption{Multi-label classification results on Flickr dataset w.r.t. Micro-f1(\%) and Macro-f1(\%).}
\centering
\setlength{\tabcolsep}{1.5mm}
\renewcommand\arraystretch{1.2}
\begin{tabular}{c|c|cc|cc|cc}
\hline
\multirow{2}{*}{\textbf{Category}}     & \multirow{2}{*}{\textbf{Methods}} & \multicolumn{2}{c|}{\textbf{0.3}} & \multicolumn{2}{c|}{\textbf{0.5}} & \multicolumn{2}{c}{\textbf{0.7}} \\ \cline{3-8} 
                                      &                                   & \textbf{Micro}  & \textbf{Macro}  & \textbf{Micro}  & \textbf{Macro}  & \textbf{Micro}  & \textbf{Macro}  \\ \hline \hline
\multirow{3}{*}{\textbf{Single-View}}  & Deepwalk                          & 34.7            & 9.8             & 34.4            & 8.1             & 34.0            & 7.3             \\
                                      & NetMF                             & 34.4            & 6.9             & 34.3            & 6.9             & 34.7            & 7.2            \\
                                      & BoostNE                           & 36.4            & 9.2             & 39.1            & 11.0            & 42.8            & 14.0            \\ \hline
\multirow{2}{*}{\textbf{Heterogeneous}}& PTE                               & 36.5            & 9.7             & 37.1            & 10.1            & 38.7            & 10.7            \\
                                      & Metapath2vec                      & 35.8            & 9.4             & 36.6            & 9.8             & 38.2            & 10.3               \\ \hline                                   
\multirow{7}{*}{\textbf{Multi-View}}   & NMF-con                           & 36.6            & \textbf{12.1}   & 38.3            & \textbf{12.6}   & 43.0            & 14.0            \\
                                      & MNE                               & 34.8            & 7.3             & 37.3            & 9.3             & 40.5            & 11.8            \\
                                      & MultiNMF                          & 35.3            & 7.7             & 37.8            & 9.7             & 40.8            & 12.1             \\
                                      & TuckER                            & 35.9            & 9.7             & 37.4            & 10.2            & 42.1            & 13.9            \\
                                      & MTNE-C                            & 36.4            & 8.3             & 37.0            & 8.0             & 39.8            & 12.5            \\
                                      & MVE                               & 36.3            & 9.6             & 35.8            & 8.5             & 37.2            & 9.2             \\ \cline{2-8} 
                                      & MNMF                              & \textbf{37.4}   & 9.8             & \textbf{39.5}   & 11.3            & \textbf{43.7}   & \textbf{14.5}   \\
                                      & (Gain)                            & +0.8            & -2.3            & +0.4            & -1.3            & +0.7            & +0.5            \\ \hline
\end{tabular}
\label{flickr_cls}
\end{table}

\subsubsection{Node Classification} The experimental results are listed in Tables~\ref{aminer_cls},~\ref{aminer_large_cls},~\ref{ppi_cls} and ~\ref{flickr_cls}. We compare our methods with three categories of baselines and show the performance gain over the best baseline. Some of the observations are listed as follows:

\noindent (1) For single-view based baseline approaches, the performance of deepwalk is barely satisfying in all datasets. But the matrix factorization based methods, i.e. NetMF and BoostNE, are much better, since explicitly factorizing the deepwalk closed-form matrices captures more network information. Thus it is better to choose matrix factorization for network embedding.

\noindent (2) Compared with state-of-the-art single-view baselines, exitsing methods for multi-view networks are not superior to them. Although multi-view based methods are much better than classic single-view network embedding method, i.e. Deepwalk, they only achieve similar performance with NetMF and BoostNE on three multi-view networks. Since existing multi-view based methods fail to capture the complex correlations between different views, they tend to generate sub-optimal embeddings for multi-view networks through linear combination strategies, such as average and concatenation. Thus it is reasonable to find that multi-view based baselines are not always better than single-view baselines.

\noindent (3) The heterogeneous network embedding methods are also not comparable to MNMF model. In summary, our model achieves a sustainable performance gain over PTE and Metapath2vec in all datasets, suggesting that MNMF is more suitable to capture the information in multi-view networks.  Since there exists only one type of node in the network, which leads to the semantic information of the meta-paths greatly reduced. Therefore, Metapath2vec is also not suitable for embedding multi-view networks. The heterogeneity of PTE comes from the text networks. In essence, the original input of PTE is a word, and the output is the embedding of each word, not multiple types of relationships between objects, thus PTE is also not suitable for representing heterogeneous edges in networks. 

\noindent (4) Compared with all three types of baselines, MNMF model always achieves better performance on four datasets with the variation of training ratio since it learns a robust representation for multi-view networks. By utilizing multi-stage NMF, the MNMF model preserves network information in different granularities with a small approximation error, thus it has around 5\% performance gain than naive matrix factorization methods on two AMiner datasets, such as NMF-con. Besides, it preserves consensus as well as unique information in multiple views, making it significantly outperforming those baselines that only preserve consensus information, such as MultiNMF. The tensor factorization method, i.e. TuckER, is also not comparable to our method since it still ignores the unique information of each view. Note that even MTNE-C also tries to preserve consensus and unique information but there still exists a gap with our model, since it fails to take high-order proximity into consideration. As a result, we conclude that the MNMF model is able to capture more complete network information.

\noindent (5) As we can see, NMF-con is almost superior to all baselines in our experimental results, and it even achieves better macro-f1 score than the proposed MNMF model in Flickr dataset. The reason goes down to the following points. First, compared with NMF-con, all baselines and MNMF only get a $d$ dimensional representation while NMF-con gets a $k*d$ dimensional representation. In general, the higher dimensions of network embedding, the better quality it has (see Fig~\ref{parameter_dims}). Thus NMF-con outperforms other baselines, but it is still not better than MNMF in most cases, which illustrates the effectiveness of MNMF. Second, the NMF-con is inefficient for network embedding, since it does not always get a low dimensional representation. For example, in PPI dataset containing six views, it would obtain a relative high ($6*256$) dimensional network embedding, which violates the original intention of network embedding. Compared with that, our MNMF model carefully assigns total dimensions into each stage, making the final representations obtained by concatenation still be low dimensional. In each stage, although a low-dimensional matrix may not be able to fit the residual matrix perfectly, but the approximation error will be regarded as another residual matrix to be decomposed in the next stage. With the stages go deep, our MNMF model even achieves lower approximation error than one-round matrix factorization, which will be illustrated in part ~\ref{whyh}.

\begin{table}[t]
\topcaption{Node clustering results w.r.t. normalized (adjusted) mutual information score.}
\centering
\setlength{\tabcolsep}{1.5mm}
\renewcommand\arraystretch{1.2}
\begin{tabular}{c|c|cc|cc|cc}
\hline
\multirow{2}{*}{\textbf{Category}}     & \multirow{2}{*}{\textbf{Method}} & \multicolumn{2}{c|}{\textbf{AMiner\_small}} & \multicolumn{2}{c|}{\textbf{AMiner\_large}} & \multicolumn{2}{c}{\textbf{PPI}} \\ \cline{3-8} 
                                      &                                  & \textbf{NMI}         & \textbf{AMI}        & \textbf{NMI}        & \textbf{AMI}        & \textbf{NMI}    & \textbf{AMI}    \\ \hline \hline
\multirow{3}{*}{\textbf{Single-view}}  & Deepwalk                         & 0.242                & 0.241               & 0.253               & 0.252               & 0.140           & 0.064           \\
                                      & NetMF                            & 0.343                & 0.341               & 0.209               & 0.205               & 0.292           & 0.178           \\
                                      & BoostNE                          & 0.368                & 0.363               & 0.330               & 0.320               & \textbf{0.300}  & 0.186           \\ \hline
\multirow{2}{*}{\textbf{Heterogeneous}}& PTE                              & 0.282                & 0.281               & 0.107               & 0.105               & 0.187           & 0.104           \\
                                      & Metapath2vec                     & 0.299                & 0.297               & 0.226               & 0.225               & 0.199           & 0.200               \\ \hline
\multirow{7}{*}{\textbf{Multi-View}}   & NMF-con                          & 0.386                & 0.382               & 0.175               & 0.163               & 0.266           & 0.208           \\
                                      & MNE                              & 0.295                & 0.291               & 0.275               & 0.262               & 0.192           & 0.105           \\
                                      & MultiNMF                         & 0.208                & 0.152               & 0.100               & 0.100               & 0.200           & 0.111           \\
                                      & TuckER                           & 0.382                & 0.377               & 0.321               & 0.311               & 0.292           & 0.220           \\
                                      & MTNE-C                           & 0.174                & 0.173               & 0.107               & 0.106               & 0.184           & 0.097           \\
                                      & MVE                              & 0.309                & 0.307               & 0.274               & 0.269               & 0.211           & 0.137           \\ \cline{2-8} 
                                      & MNMF                             & \textbf{0.409}       & \textbf{0.408}      & \textbf{0.330}      & \textbf{0.324}      & 0.291           & \textbf{0.231}  \\
                                      & (Gain \%)                         & +2.3                 & +2.6                & +0.0                & +0.4                & -0.9            & +1.1            \\ \hline
\end{tabular}
\label{cluster}
\end{table}

\subsubsection{Node Clustering}
The node clustering results on AMiner and PPI networks can be observed from Table.~\ref{cluster}. In AMiner networks, our MNMF model achieves best performance over all baselines. Compared with the performance of baselines, MNMF has 2.3\% and 2.6\% gain in terms of NMI and AMI respectively in the small network, which indicates the embedding learned by MNMF model can well preserve the information in original network. For the large network, existing methods such as MultiNMF, NMF-con and MTNE even cause degradation of performance compared with single view based methods. MNMF model still outperforms all baselines but the performance gain is not so noticeable. Apart from AMiner network, as the view number is relative large in PPI network, it is challenging to integrate information in multiple views. Thus exitsing multi-view based methods fail to get satisfactory result, but MNMF model still obtains the best AMI score in PPI network. This phenomenon shows that MNMF model is more effective than existing works when learning representations for complex multi-view networks. 

\begin{figure}[t]
    \centering
    \subfigure[Citation View]{
    \includegraphics[scale=0.0425]{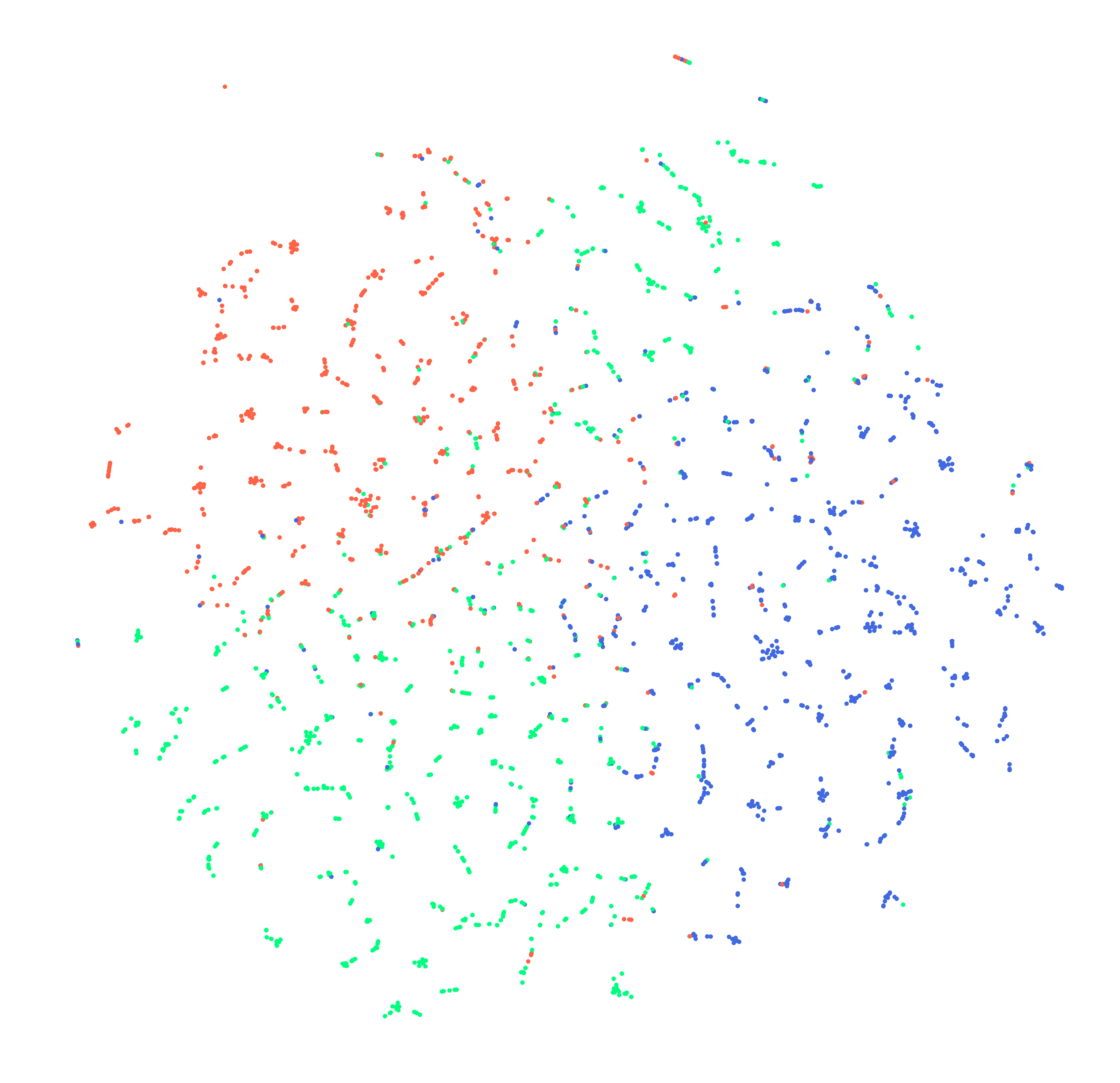}
    }
    \subfigure[Coauthor View]{
    \includegraphics[scale=0.0425]{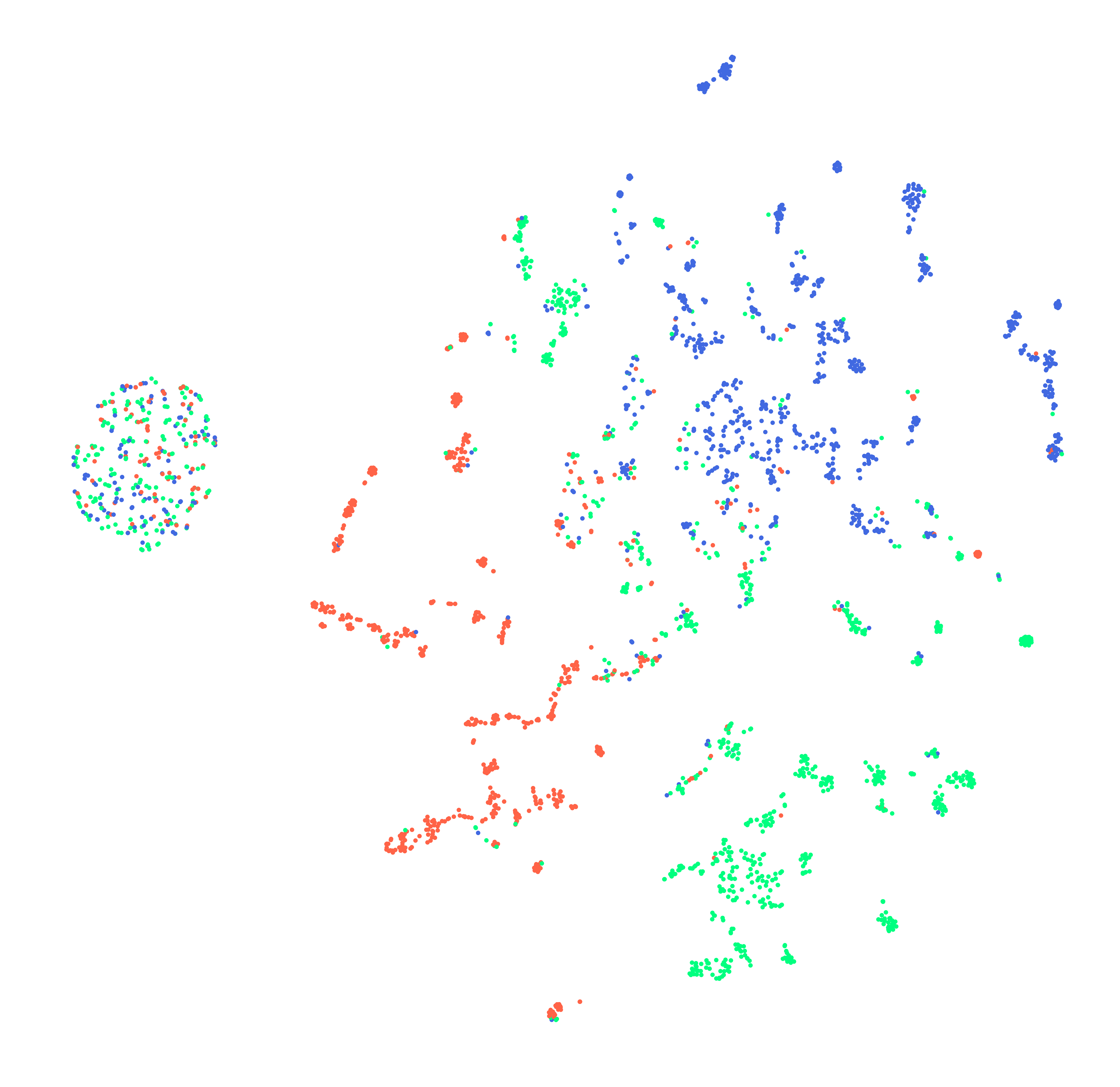}
    }
    \subfigure[Similarity View]{
    \includegraphics[scale=0.0425]{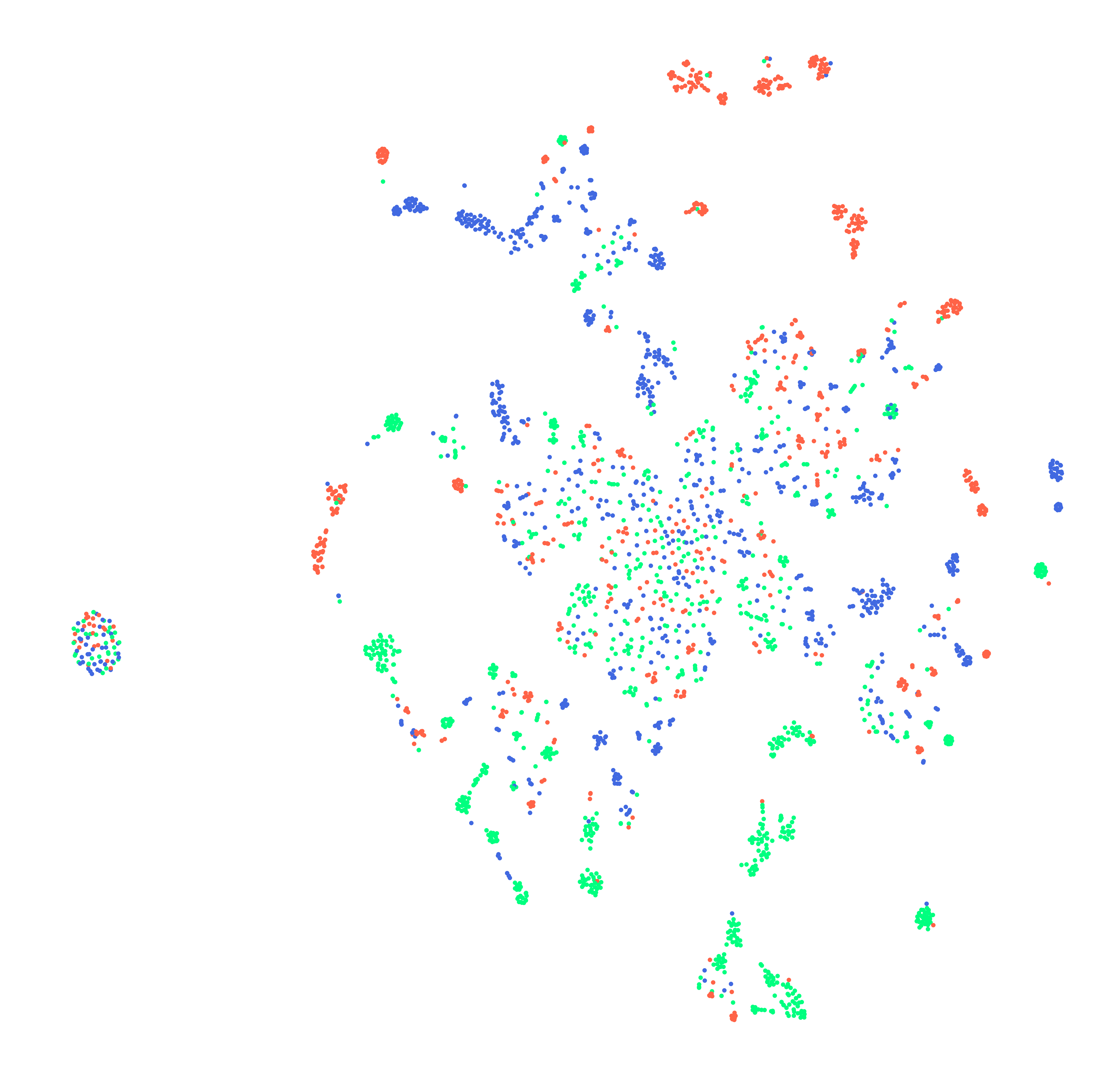}
    }
    \subfigure[stages=2]{
    \includegraphics[scale=0.0425]{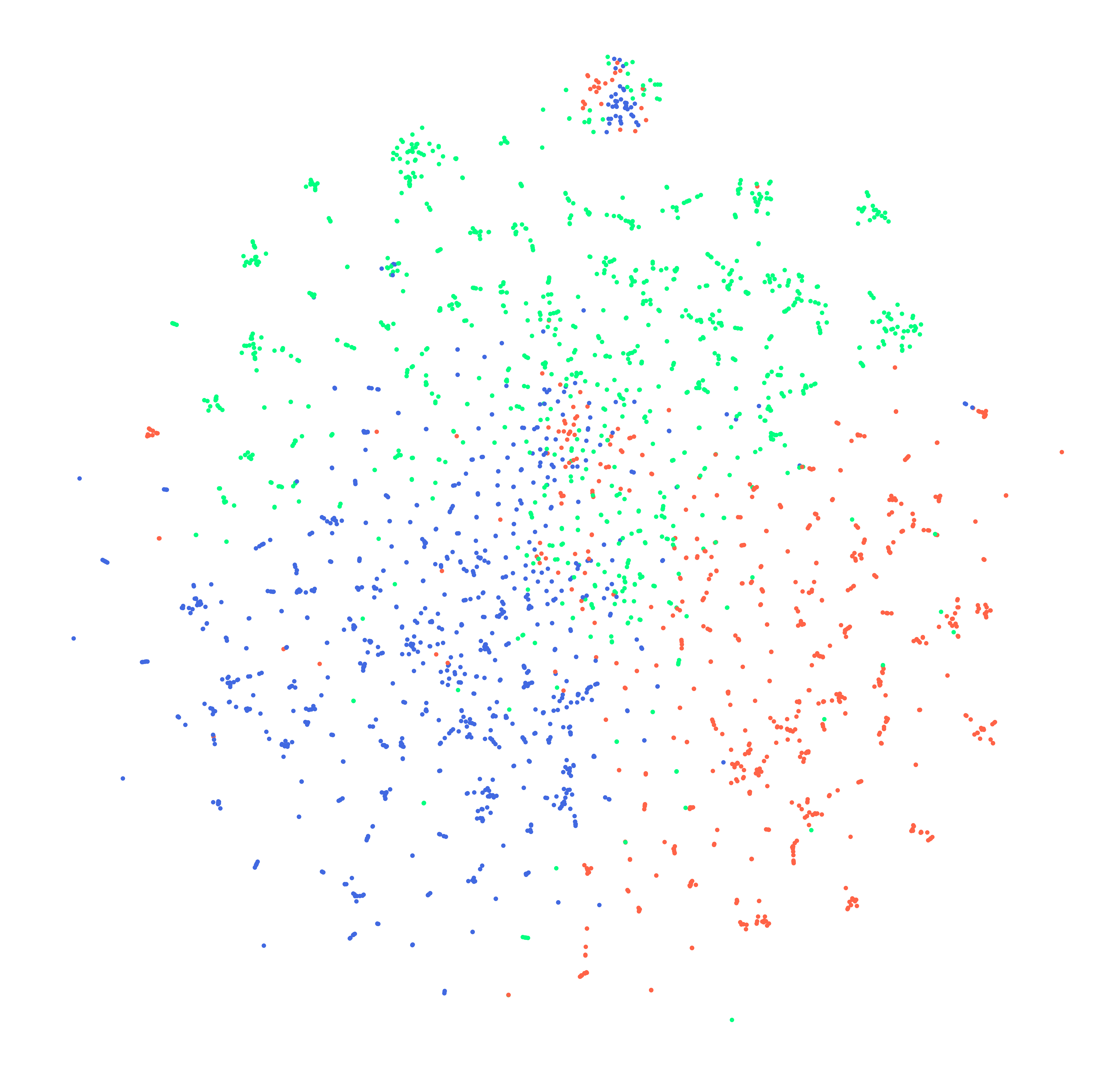}
    }
    \subfigure[stages=4]{
    \includegraphics[scale=0.0425]{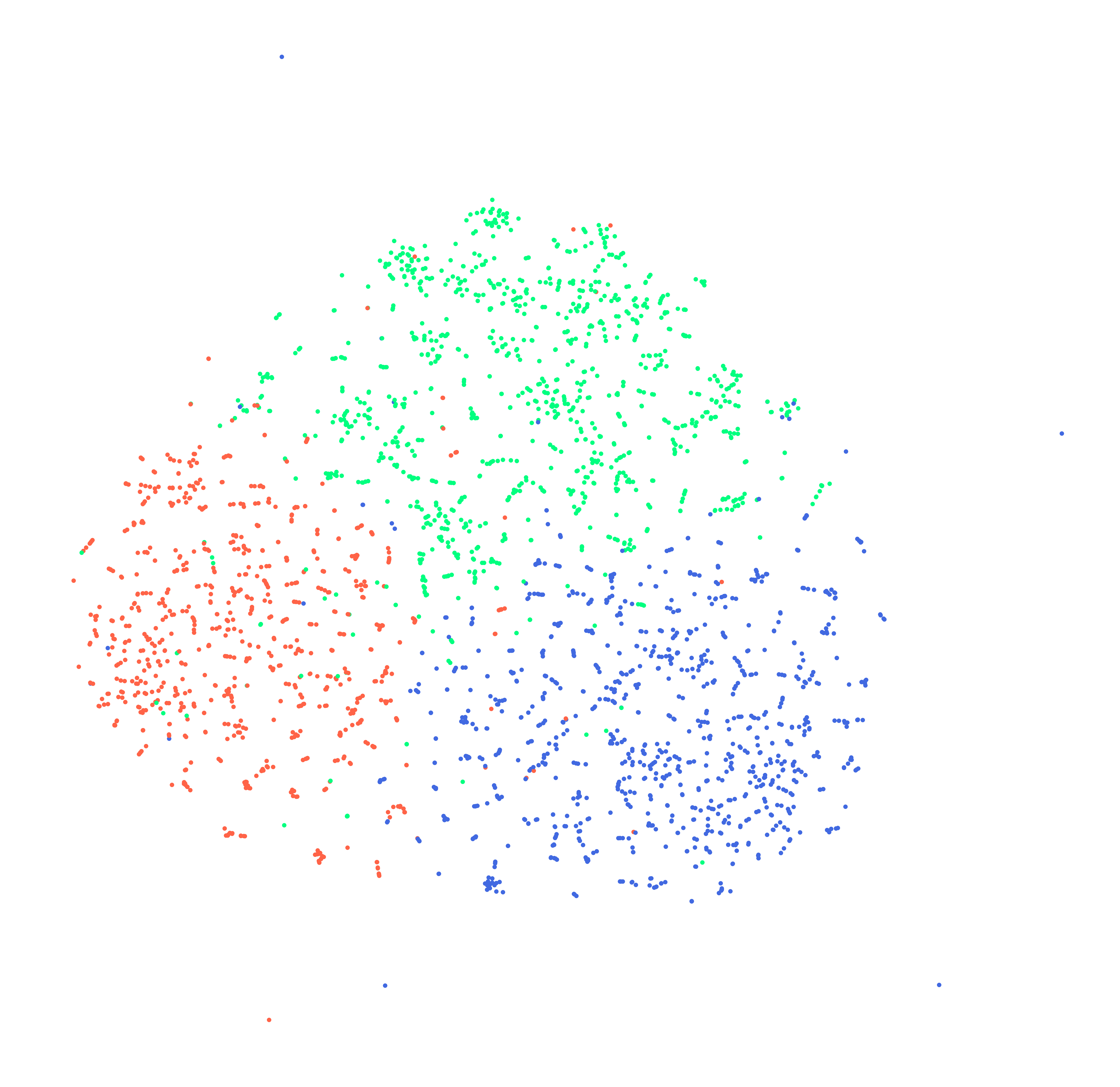}
    }
    \subfigure[stages=8]{
    \includegraphics[scale=0.0425]{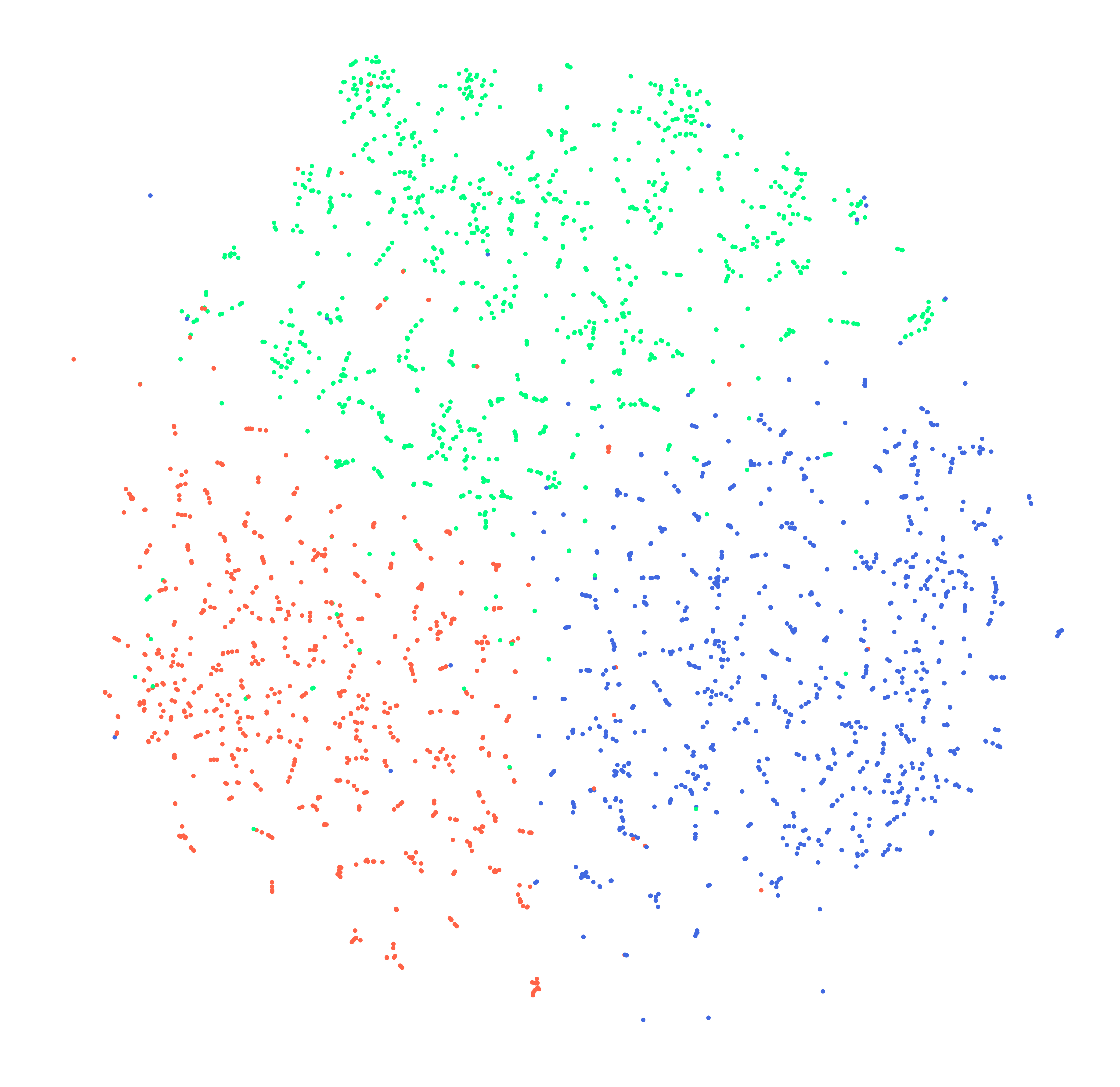}
    }
    \caption{2d t-SNE Visualization for AMiner\_small Dataset. Each point represents a node and colors represent labels. red : computational linguistics; blue : computer graphics; green : theoretical computer science.}
    \label{vis}
\end{figure}

\subsubsection{Network Visualization} We project the learned representations for AMiner dataset into a 2 dimensional vectors with t-SNE method \cite{maaten2008visualizing}. Fig~\ref{vis} shows the visualization results by the deepwalk matrix of each single view (a, b, c) and different stages' visualization results (d, e, f) of the MNMF model, and each color represents a research field. We can see that our multi-view based results are much better than single view based results. Citation view fails to clearly separate the researchers from three research fields. Both coauthor and text similarity views form an independent community that does not actually exist. This observation proves that it is necessary to learn a representation by utilizing multiple views since single view is not enough to describe the network precisely. By carefully learning a robust representation from multiple views, our method can well distinguish three research fields when the number of stages is large enough. But we still notice that if the stage is too low (e.g. stage = 2 and 4) the visualization results are still not so satisfactory. This phenomenon is due to that low stage factorized matrices fail to preserve complete information of original matrices with a large approximation error. With the stage increasing, we see a significant improvement in the visible results. Based on the above discussion, we can conclude that our method can better capture the information of multi-view network through multi-stage matrix factorization.

\begin{figure}[t]
    \centering
    \subfigure[\scriptsize{AMiner\_small}]
    {
    \includegraphics[scale=0.205]{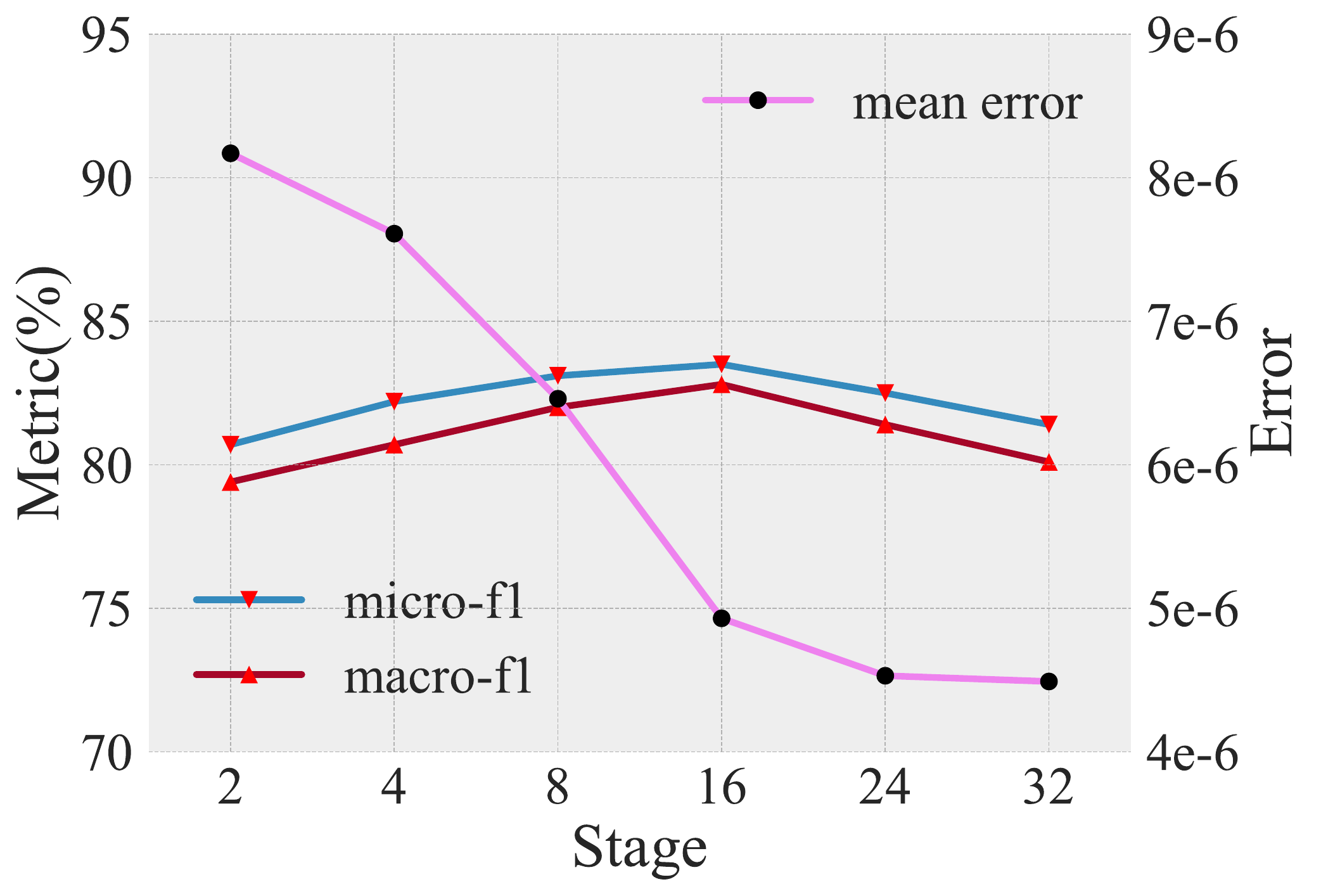}
    }
    \subfigure[\scriptsize{PPI}]
    {
    \includegraphics[scale=0.205]{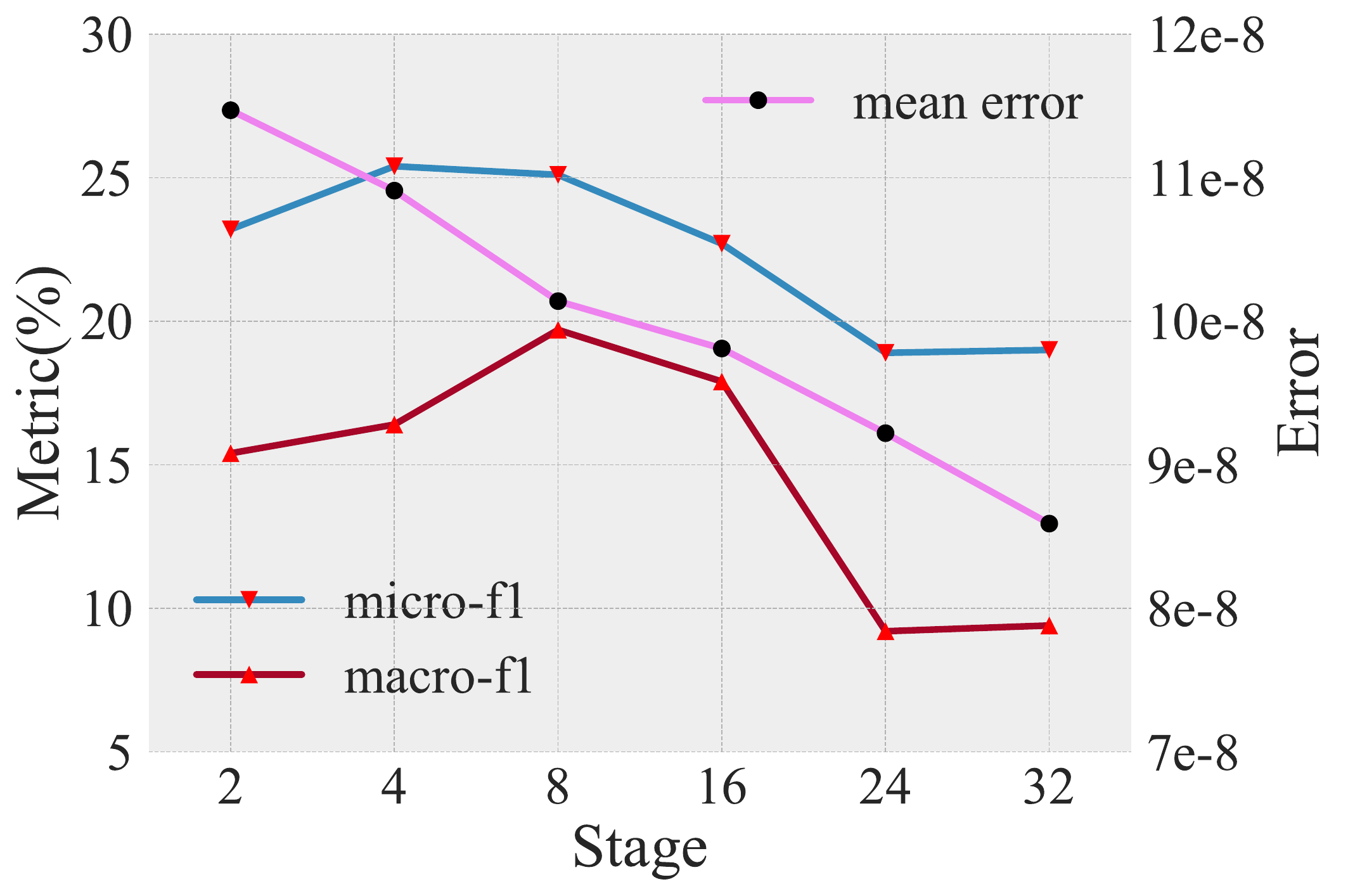}
    }
    \subfigure[\scriptsize{Flickr}]
    {
    \includegraphics[scale=0.205]{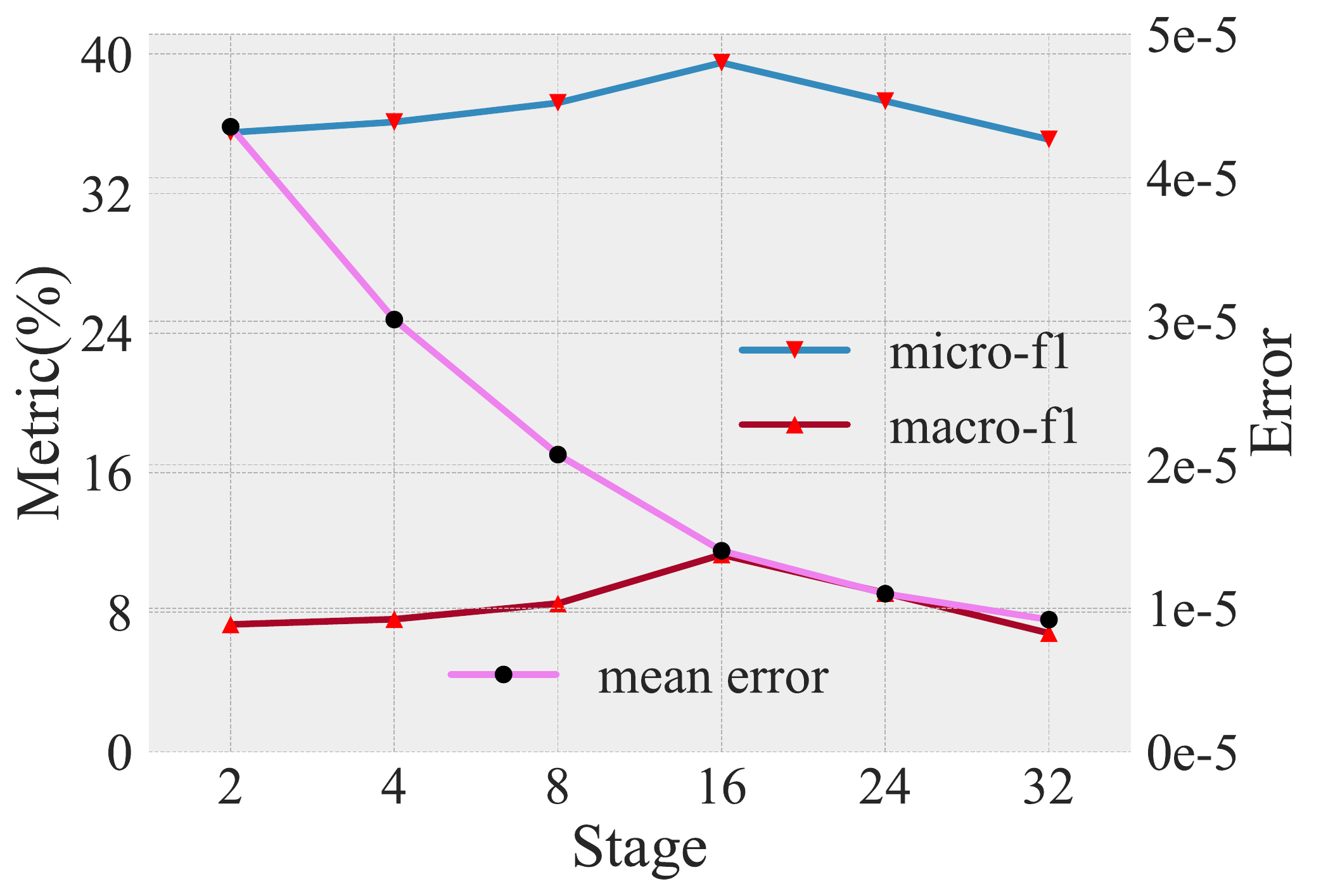}
    }
    \caption{Node classification performance and mean approximation error of MNMF model w.r.t. the number of stages.}
    \label{hierarchy}
\end{figure}

\begin{figure}[t]
    \centering
    \subfigure[Consensus and unique information between views]{
    \includegraphics[scale=0.35]{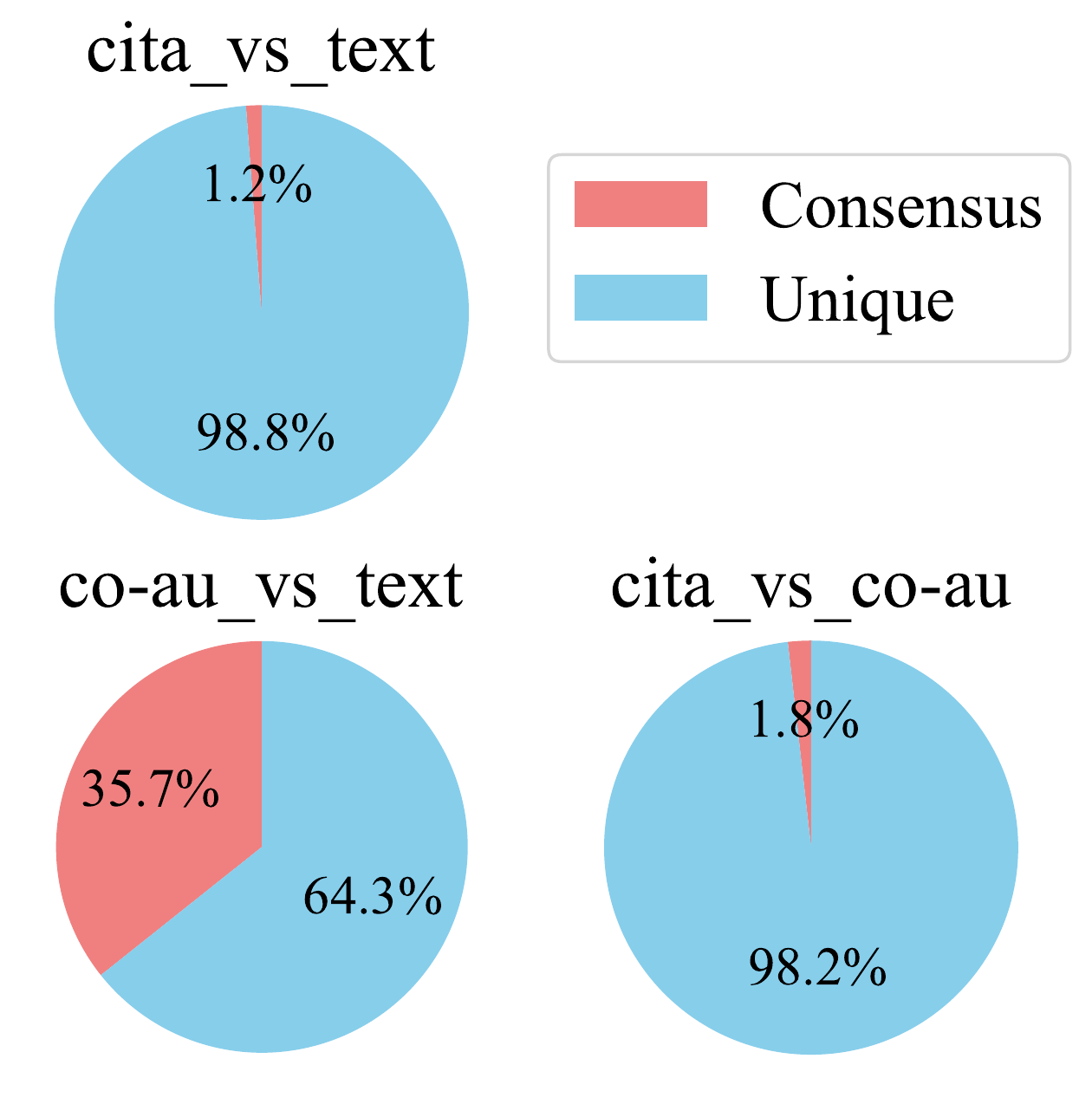}
    }
    \subfigure[Ablation test]{
    \includegraphics[scale=0.35]{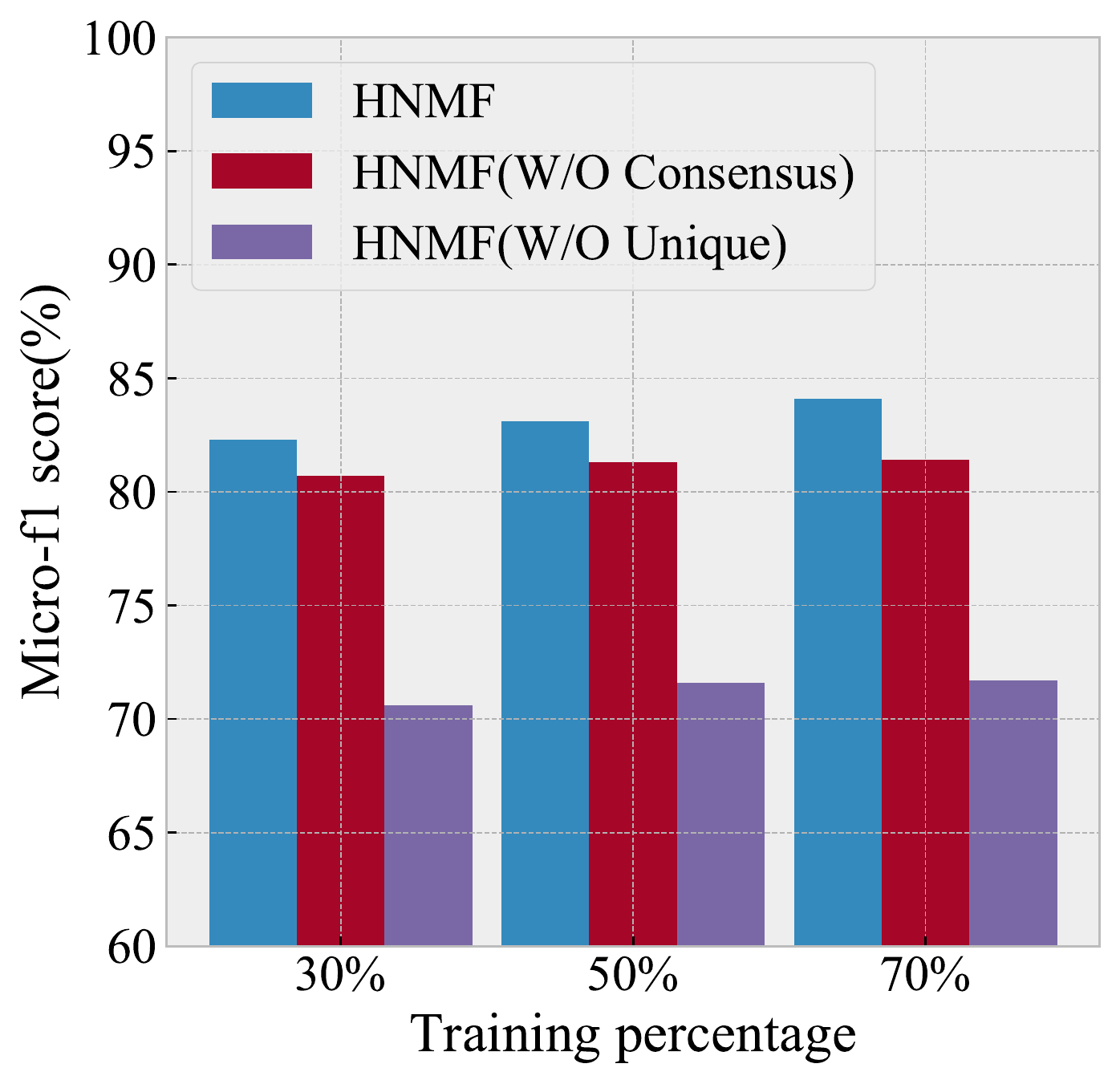}
    }
    \caption{Case study on AMiner\_small Network. \small cita: citation view; co-au: co-authorship view; text: text similarity view.}
    \label{ablation}
\end{figure}

\subsection{Why Multi-stage Matrix Factorization? (Q2)} \label{whyh}
Although the visualization results have indicated that multi-stage matrix factorization leads to a better performance with the stage going deeper, the empirical results cannot give an intuitive explanation of this phenomenon. Here we try to investigate the underlying reason from the  perspective of approximation error using experiments on AMiner\_small, PPI, and Flickr datasets with training rate as $0.5$. From Fig.\ref{hierarchy} (a), (b) and (c), we can see the classification metrics increase obviously when the stage goes deep at the beginning, at the same time the mean approximation error also reduces sharply. According to the matrix theory, matrix factorization approaches are built on a condition: the matrix $\bm{X} \in \mathbb{R}^{N*N}$ to be decomposed should be of low-rank, i.e. its rank should be no more than the given dimension $d$ , so that the one-round matrix factorization process is possible to generate a smaller dimensional matrix $\bm{U} \in \mathbb{R}^{N*d}$, which fits all the information of the original matrix. This demand is difficult to meet in practice especially for multi-view networks, since it is strait to ensure that each view satisfies the low-rank condition. Thus it is common that the one round matrix decomposition process has a large error, so we perform matrix factorization process in a forward manner on the residual matrices. When the stage goes deeper the residual matrices will become sparser, the approximation error would reduce gradually. As a result, our method achieves better performance than compared baselines. 

Apart from that, we still find that if the stages go very deep, such as 24 and 32, the node classification performance declines pretty clearly in all datasets. We analyze the inherent reason as follows: Observing from Fig.~\ref{hierarchy}, at the beginning, with the stage going deep, the performance will increase gradually. Then, at a certain stage (16 for AMiner\_small and Flickr), if we further increase the stages, the reduction of mean approximation error is very limited. All the above discussion can reflect that the multi-stage structure is helpful only when the number of stages are selected carefully.

\begin{table}[ht]
\centering
\caption{Tune the number of stages for MNMF.}
\setlength{\tabcolsep}{1.5mm}
\renewcommand\arraystretch{1.2}
\label{tune-stage}
\begin{tabular}{c|ccc|ccc}
\hline
\multirow{2}{*}{\textbf{Method}} & \multicolumn{3}{c|}{\textbf{AMiner\_small}}          & \multicolumn{3}{c}{\textbf{PPI}}                      \\ \cline{2-7} 
                                 & \textbf{0.3}  & \textbf{0.5}  & \textbf{0.7}  & \textbf{0.3}  & \textbf{0.5}  & \textbf{0.7}  \\ \hline \hline
Baseline                         & 81.6          & 82.6          & 83.5          & 19.6          & 22.4          & 23.6          \\
Grid Search                      & \textbf{82.1} & \textbf{83.2} & \textbf{84.2} & \textbf{20.3} & \textbf{23.7} & \textbf{24.8} \\
Bayesian Search                  & 82.0          & 83.1          & 84.1          & 20.3          & 23.6          & 24.6          \\
Early Stopping                   & 81.8          & 82.9          & 83.9          & 20.1          & 22.8          & 24.2          \\ \hline
\end{tabular}
\end{table}

In order to decide a suitable number of stages, we suggest three simple yet effective methods to determine the number of stages. First, we can do a grid parameter search, we can determine several values to try, and then traverse all the parameter values to choose the best number of stage. The advantage of this method is simple and reliable if we try enough candidate values, but it may only be feasible in some small datasets due to the bad time cost. However, motivated by the advantages of auto machine learning, we can also choose Bayesian optimization~\cite{snoek2012practical}, which takes into account the corresponding experimental results of different parameters to save more time. Second, motivated by early-stopping, we can also choose a threshold to stop the multi-stage process in advance. For example, we can perform an evaluation on a small part of samples to determine whether stopping the multi-stage process when the metric is lower than last stage. To validate these two method, we also perform an experiment on AMiner\_small dataset and PPI dataset and the node classification results~(average of micro-F1 score and macro-F1 score) are shown in Table.~\ref{tune-stage}. We set the search space of the number of stages is from 2 to 64 and choose $10\%$ samples as a validation set. As we can see, the grid search achieves the best results since it traverses all possible values; the bayesian search and early stopping also achieve better performance than baseline.

\subsection{Analysis of View Weights (Q2)} \label{view_w}
In this subsection, we explore the view weights learned by the MNMF model to validate the effectiveness of the weighting strategy. We select two datasets as examples to check the view weights~(i.e. $\alpha$) with $\gamma$ as 2.  As shown in Table~\ref{view_weight_2}, the value of $\gamma$ has a noticeable impact on the learned view weights, which reflects the importance of the weighting strategy. As we can see, the differences among weights of different view are relatively large in the beginning, for example, on AMiner\_small dataset, the weight of the first view is significantly higher than the second view, which indicates that our model will focus on the first view. With the stages going deep, the weight of each view will gradually tends to close to each other, so we see that the weight of the first and third view gradually decreases, and the weight of the second view gradually increases. This phenomenon may be due to that with the increasing of the stage number, more and more information has been retained, thus the difference between two views is getting smaller, leading their weights being closer and closer.  In accordance with that, we also observe the similar results on PPI dataset. As a result, we can conclude that the model can learn the importance of multiple views through assigning different weight to each view.

\begin{table}[t]
\centering
\topcaption{View weights at each stage}
\setlength{\tabcolsep}{1.0mm}
\renewcommand\arraystretch{1.2}
\label{view_weight_2}
\begin{tabular}{c|c|c|c|c|c|c|c|c|c}
\hline
\multicolumn{2}{c|}{\textbf{Dataset}}    & stage 1 & stage 2 & stage 3 & stage 4 & stage 5 & stage 6 & stage 7 & stage 8 \\ \hline \hline
\multirow{3}{*}{\textbf{AMiner\_small}}  & view 1 & 0.4510  & 0.4422  & 0.4343  & 0.4279  & 0.4224  & 0.4172  & 0.4120  & 0.4077  \\ 
                                         & view 2 & 0.1815  & 0.1925  & 0.2031  & 0.2111  & 0.2181  & 0.2249  & 0.2326  & 0.2380  \\
                                         & view 3 & 0.3675  & 0.3653  & 0.3626  & 0.3610  & 0.3595  & 0.3579  & 0.3554  & 0.3542  \\ \hline \hline
\multirow{6}{*}{\textbf{PPI}}            & view 1 & 0.0564  & 0.0631  & 0.0700  & 0.0769  & 0.0820  & 0.0870  & 0.0931  & 0.0985  \\
                                         & view 2 & 0.2360  & 0.2245  & 0.2157  & 0.2080  & 0.2010  & 0.1959  & 0.1903  & 0.1849  \\
                                         & view 3 & 0.1683  & 0.1836  & 0.1961  & 0.2056  & 0.2145  & 0.2187  & 0.2219  & 0.2254  \\
                                         & view 4 & 0.0879  & 0.0925  & 0.0952  & 0.0994  & 0.1042  & 0.1097  & 0.1155  & 0.1216  \\
                                         & view 5 & 0.2403  & 0.2284  & 0.2188  & 0.2105  & 0.2032  & 0.1971  & 0.1911  & 0.1854  \\
                                         & view 6 & 0.2111  & 0.2080  & 0.2043  & 0.1995  & 0.1951  & 0.1916  & 0.1881  & 0.1842  \\ \hline                    
\end{tabular}
\end{table}

\subsection{Consensus or Unique Information? (Q2)} \label{consensus_vs_unique}
Here we evaluate the contribution of preserving the consensus and unique information in multiple views respectively. Firstly, we perform a statistical analysis on AMiner\_small dataset. As shown in Fig.~\ref{ablation}(a), we follow the method in \cite{shi2018mvn2vec,wang2020rgae} to study the consensus and unique information between different views.  For a pair of views, each node has two unique neighbors set in these two views. If the Jaccard coefficient between its two neighbors sets is greater than 0.5, then we think the node carries consensus information in these two views, otherwise carries unique information. The Jaccard coefficient describes the overlap of the neighborhood set for a node in two different views, which is also similar to the 2-nd order proximity in LINE~\cite{tang2015line}. In LINE, the 2-nd order proximity between two nodes is defined as the degree of overlap between the neighborhood sets of the two nodes, and a large value indicates the two nodes should be close to each other, i.e. carrying similar information. Intuitively, if a node has a high degree of neighbor coincidence in two different views, it indicates that the information of the two view is very consistent.  As we can see, there exist noticeable proportion of nodes in co-authorship and text similarity views carrying consensus information while other pairs of views are opposite. Thus we conclude that it might be inappropriate to simply preserve consensus or unique information for multi-view network embedding since the view relationships are complex. 

Then we perform an ablation test on node classification task to verify our conclusion. From Fig.~\ref{ablation}(b), we can see the Micro-f1 score of MNMF model declines after removing any type of information. If we remove consensus information from MNMF model, the performance would drop around 2\%. Compared with that,  without unique information, it brings even worse impact on MNMF model, resulting in about 12\% performance loss in terms of Micro-f1 score. This phenomenon is corresponding to the statistical analysis, different views would carry more unique information than consensus information since network data have high variance. However, we cannot ignore the consensus information since it helps to build connections between different views, thus we should preserve both consensus and unique information to learn robust representations for multi-view networks.

\begin{figure}[t]
    \centering
    \subfigure[\scriptsize{AMiner\_small}]
    {
    \includegraphics[scale=0.225]{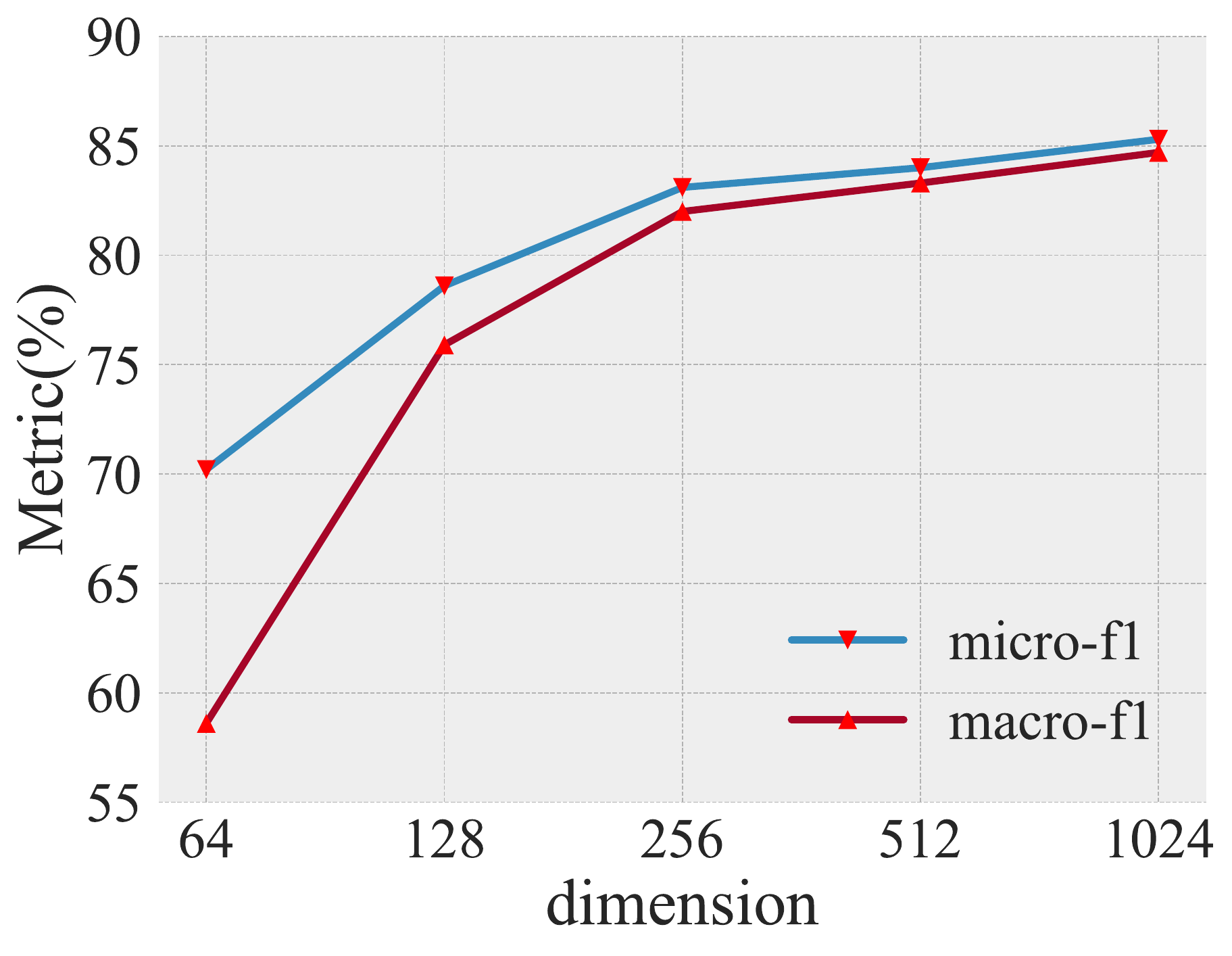}
    }
    \subfigure[\scriptsize{PPI}]
    {
    \includegraphics[scale=0.225]{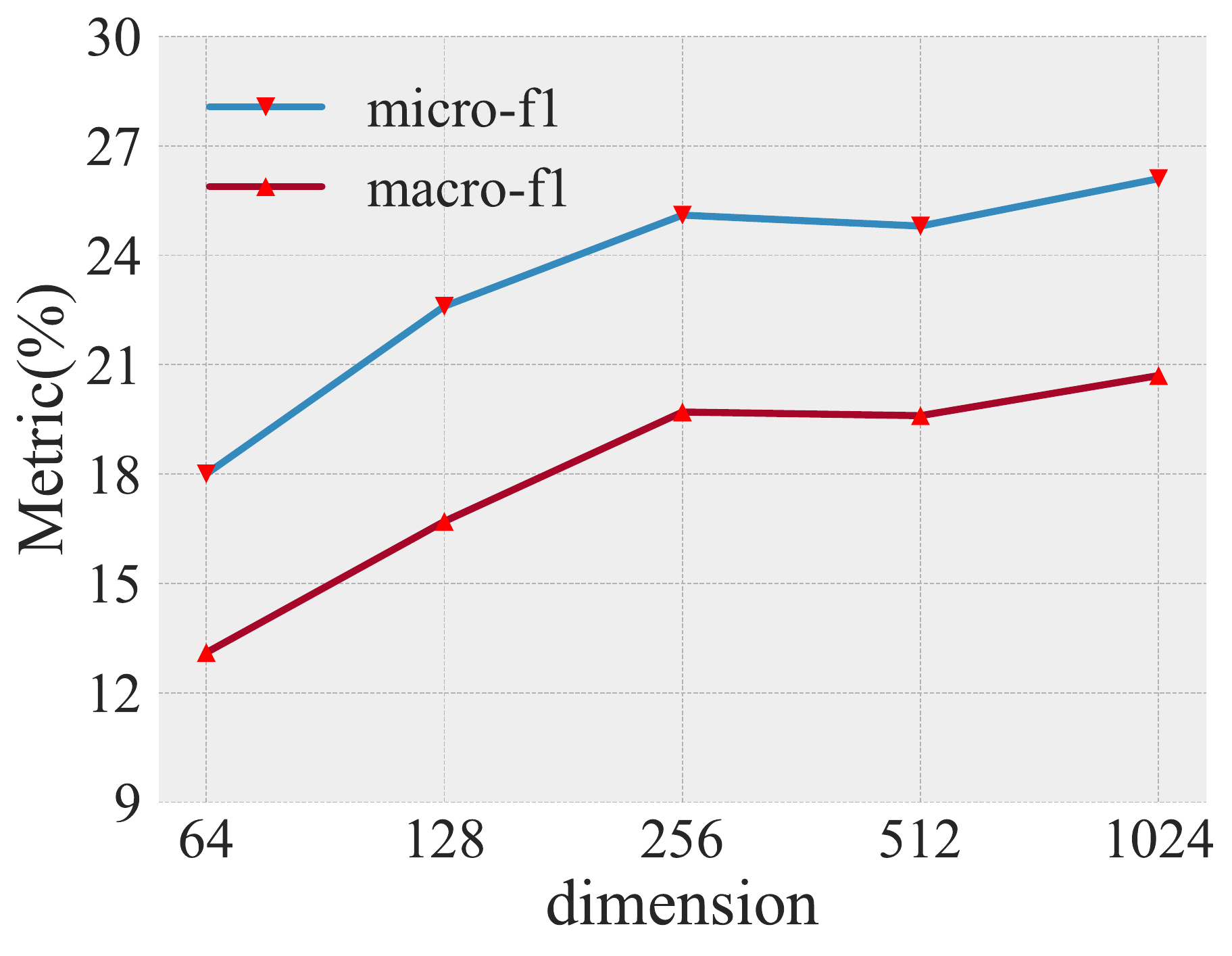}
    }
    \subfigure[\scriptsize{Flickr}]
    {
    \includegraphics[scale=0.225]{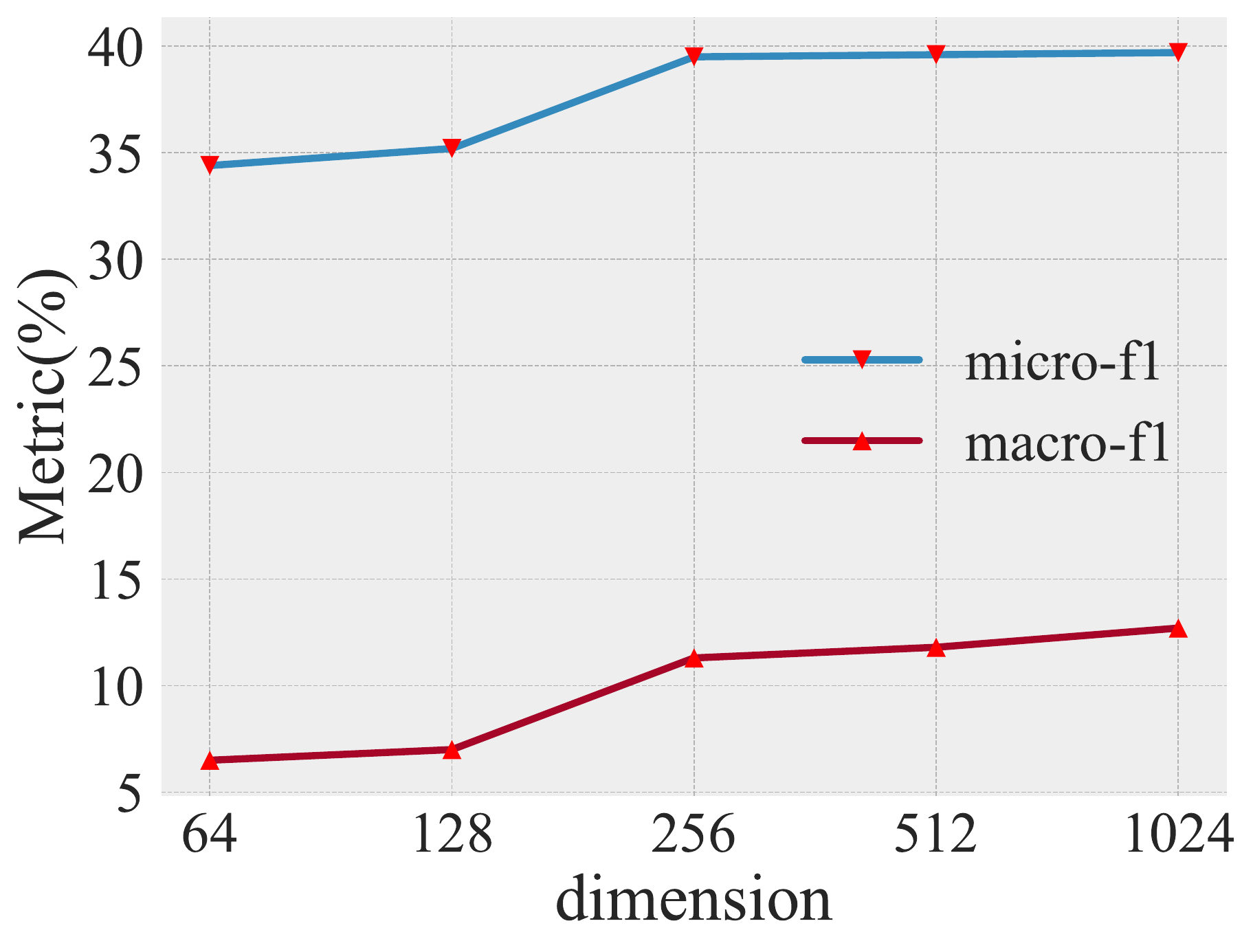}
    }
    \caption{Node classification performance of MNMF model w.r.t. the number of dimensions.}
    \label{parameter_dims}
\end{figure}

\begin{figure}[t]
    \centering
    \subfigure[\scriptsize{AMiner\_small}]
    {
    \includegraphics[scale=0.225]{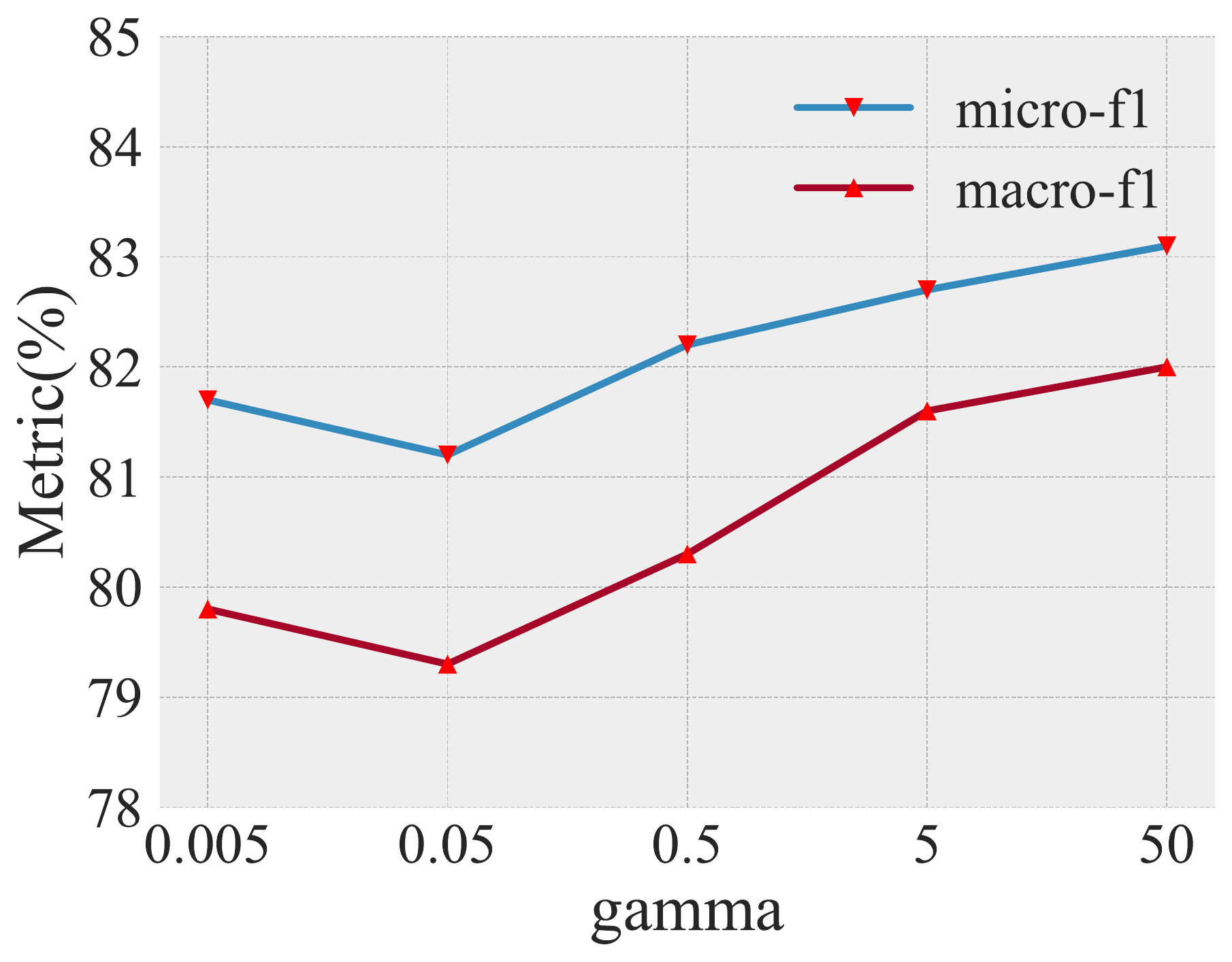}
    }
    \subfigure[\scriptsize{PPI}]
    {
    \includegraphics[scale=0.225]{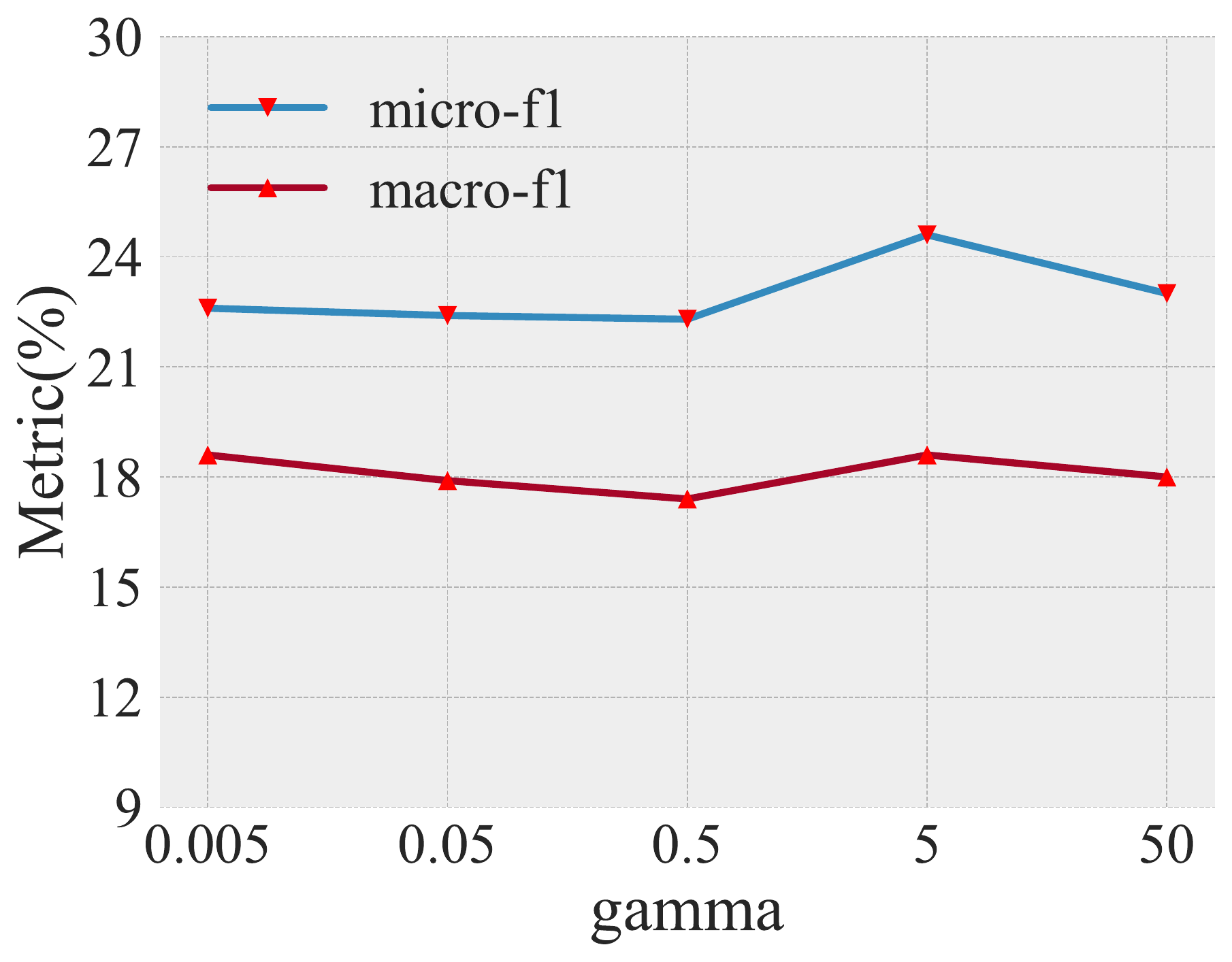}
    }
    \subfigure[\scriptsize{Flickr}]
    {
    \includegraphics[scale=0.225]{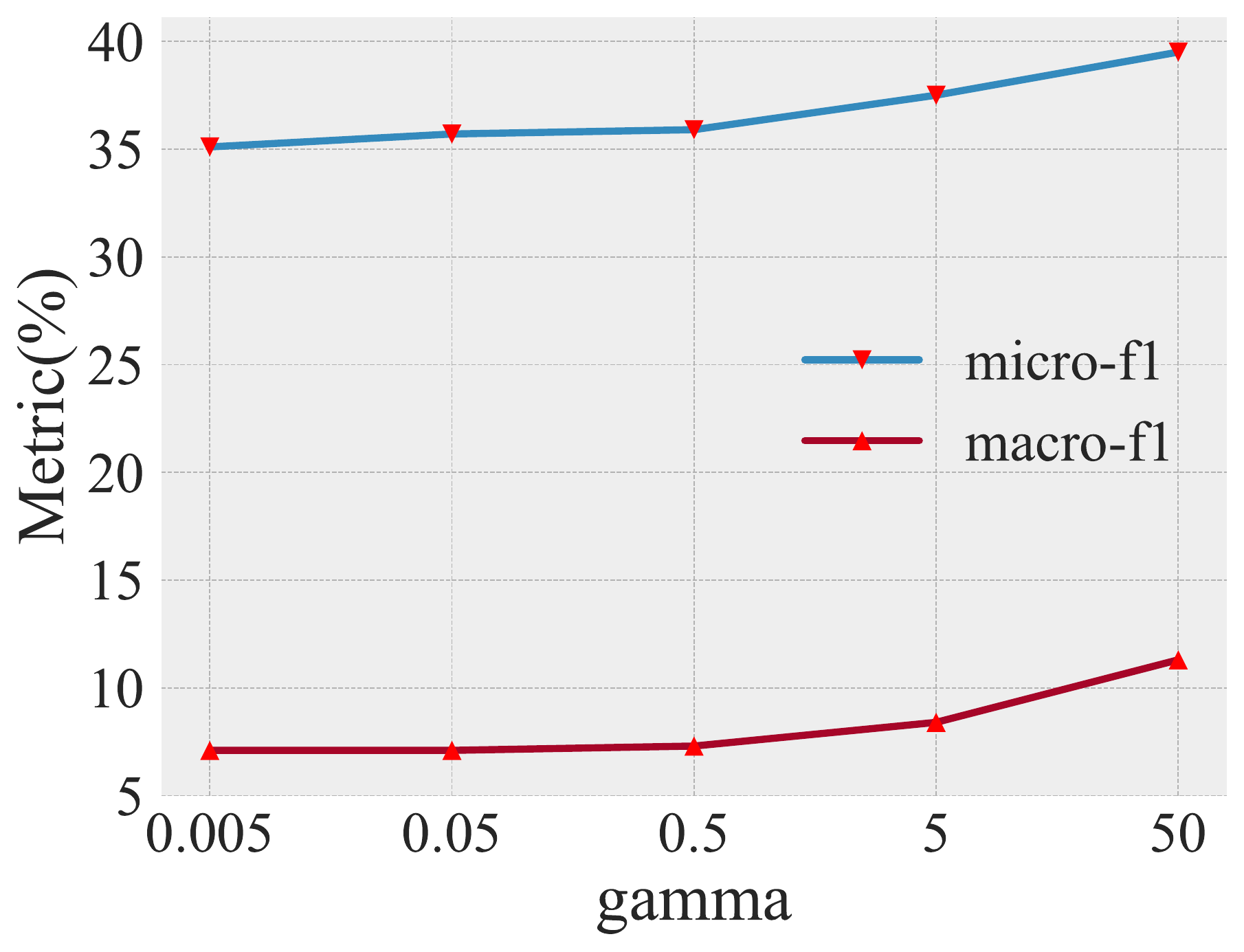}
    }
    \caption{Node classification performance of MNMF model w.r.t. the value of $\gamma$.}
    \label{parameter_gamma}
\end{figure}

\begin{figure}[t]
    \centering
    \subfigure[\scriptsize{AMiner\_small}]
    {
    \includegraphics[scale=0.225]{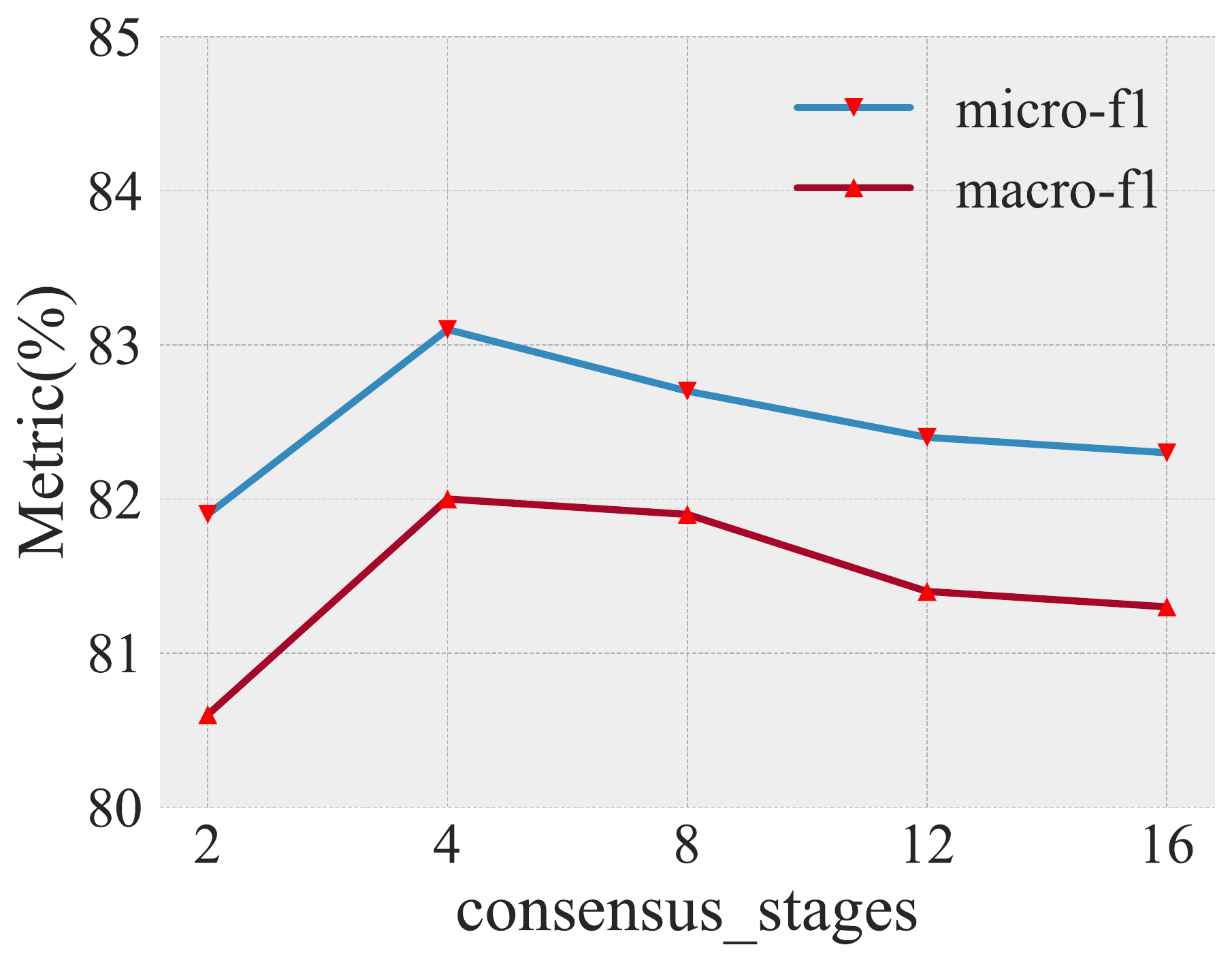}
    }
    \subfigure[\scriptsize{PPI}]
    {
    \includegraphics[scale=0.225]{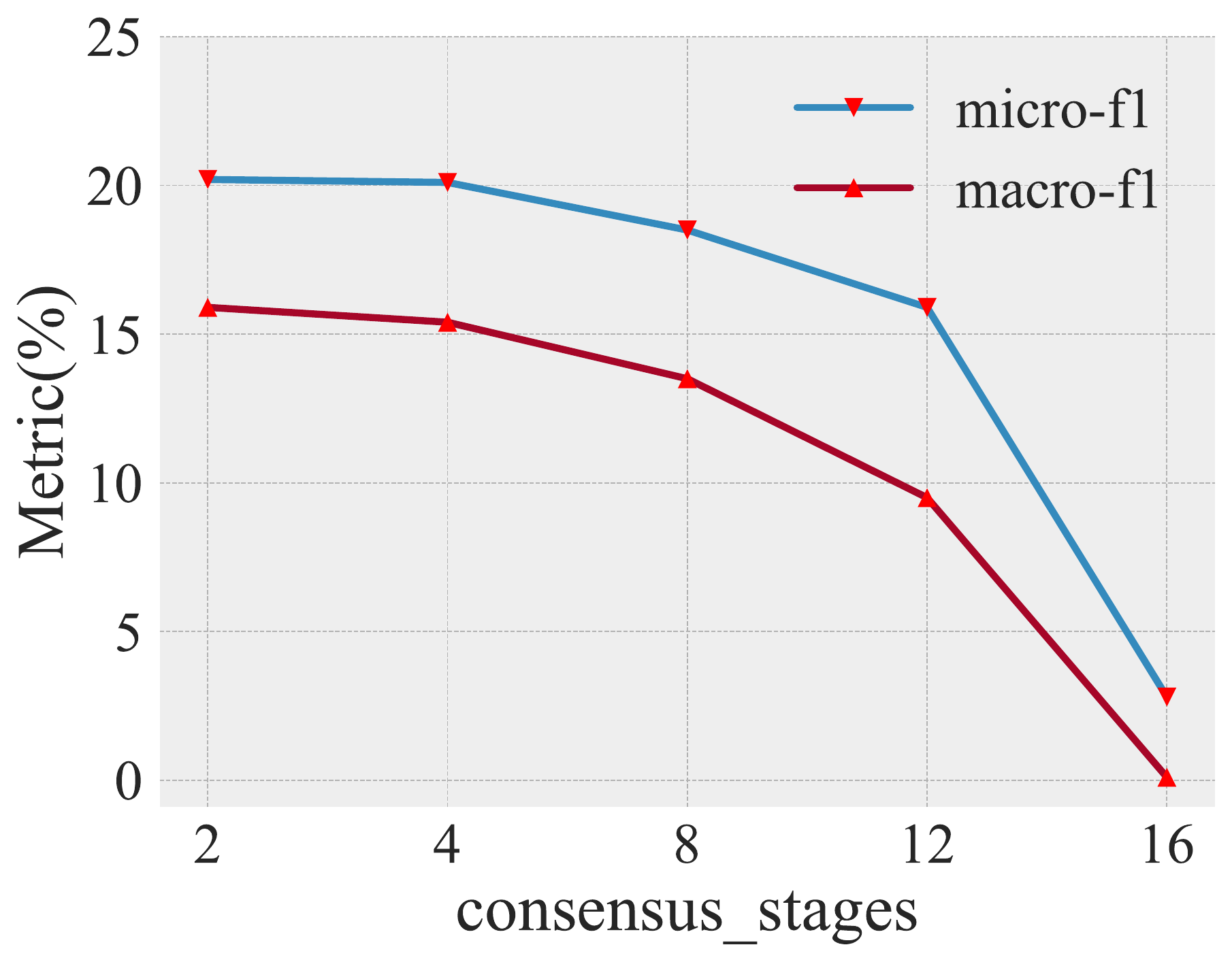}
    }
    \subfigure[\scriptsize{Flickr}]
    {
    \includegraphics[scale=0.225]{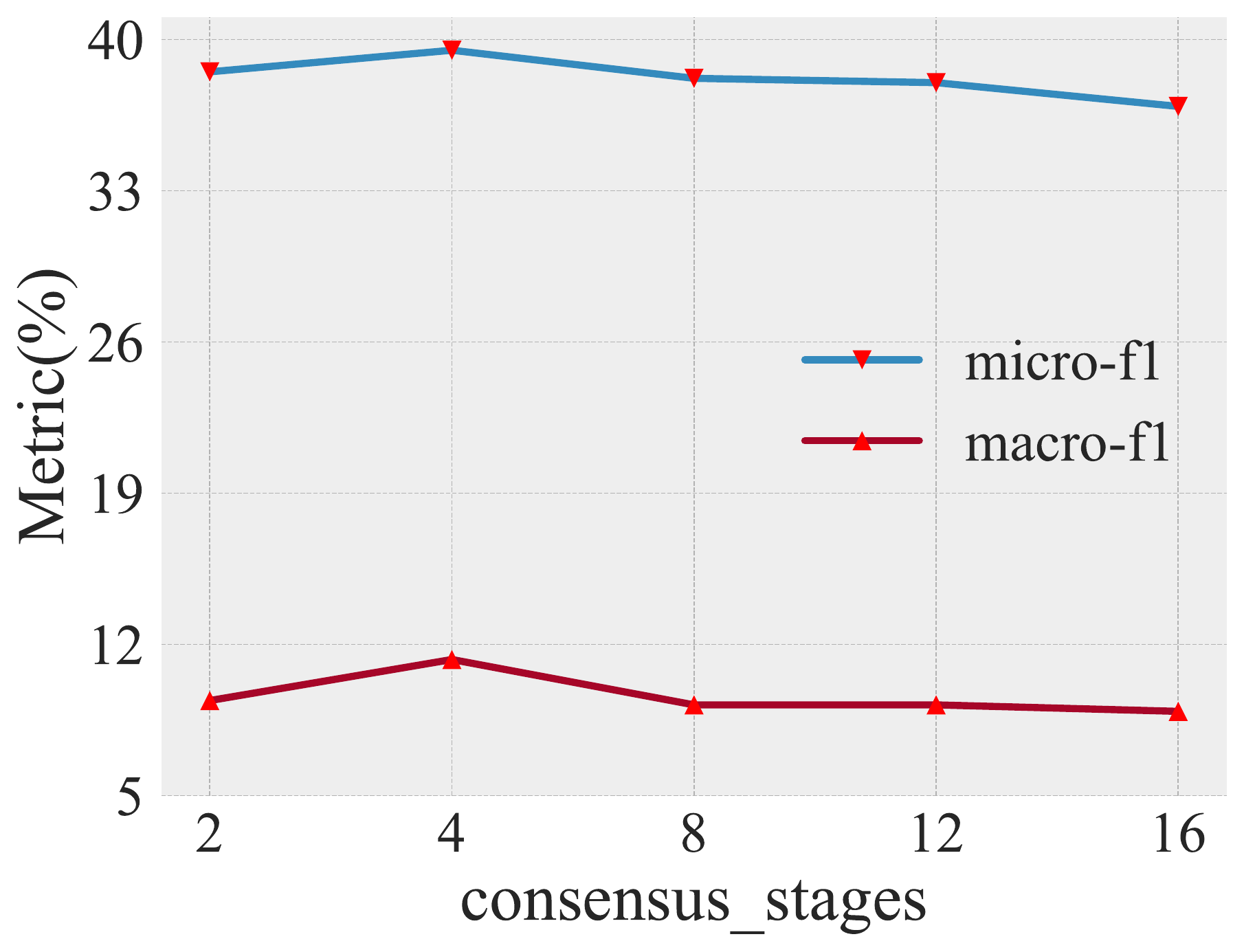}
    }
    \caption{Node classification performance of MNMF model w.r.t. the number of consensus stages.}
    \label{parameter_consensus_stages}
\end{figure}

\subsection{Parameter Sensitivity (Q3)}
Here we discuss the parameter sensitivity of MNMF model. We take AMiner\_small, PPI and Flickr datasets as examples and keep the training ratio as 0.5, aiming to investigate how the representation dimensions, the different values of $\gamma$ and the different number of consensus stages affect the performance on node classification task. We first analyze the influence of the representation dimensions. As shown in Fig.~\ref{parameter_dims}, All datasets follow the same pattern. When the dimension is too low, the performance is not satisfactory. The performance initially raises when the dimension increases, since a low dimension is unable to encode sufficient information in the latent space. However, if the dimension continuously increases, the performance begins to increase not so obviously. This is due to that too high dimension representations may overfit the network information. As a result, we conclude that we should choose an appropriate embedding dimension, such as 256. Too lower the dimension is, the embeddings are not able to preserve the complete information; Too higher the dimension is, the embeddings take up more storage space and increase the risk of overfitting.

Secondly, from Fig.~\ref{parameter_gamma} we can see that hyper-parameter $\gamma$ actually influences the results. It is efficient to use the one parameter $\gamma$ for controlling the different importance of each view during the optimization process dynamically. According to Eq.\eqref{alpha_update} we would assign equal weights to all views when $\gamma$ closes to $\infty$. When $\gamma$ closes to 1, the weight for the view whose $\bm{W}_k^l$ value is smallest will be assigned as 1, while others are almost ignored since their weights are close to 0. A small $\gamma$ will lead the poor performance especially for Macro-f1 score, thus we need to choose a suitable $\gamma$ in order to guarantee the weights being assigned to each view properly. According to Fig.~\ref{parameter_gamma}, it is more suitable to set the value of $\gamma$ larger than 5.

Finally, we examine the influence of the number of stages for preserving consensus information (corresponding to Eq.~\eqref{consensus}). We set the total number of stages as 16. As shown in Fig.~\ref{parameter_consensus_stages}, we observe consistent trends across the three datasets. As the consensus information is usually less than unique information, we should pay more attention to the unique information from all views. Compared with existing methods~\cite{lu2019auto,xu2019multi}, our MNMF model is more flexible to tune the importance of consensus information by setting different stages for preserving it. On the whole, the performance of MNMF model will first increase with the consensus stages going deep and then decrease dramatically. 

\begin{figure}
    \centering
    \includegraphics[scale=0.4]{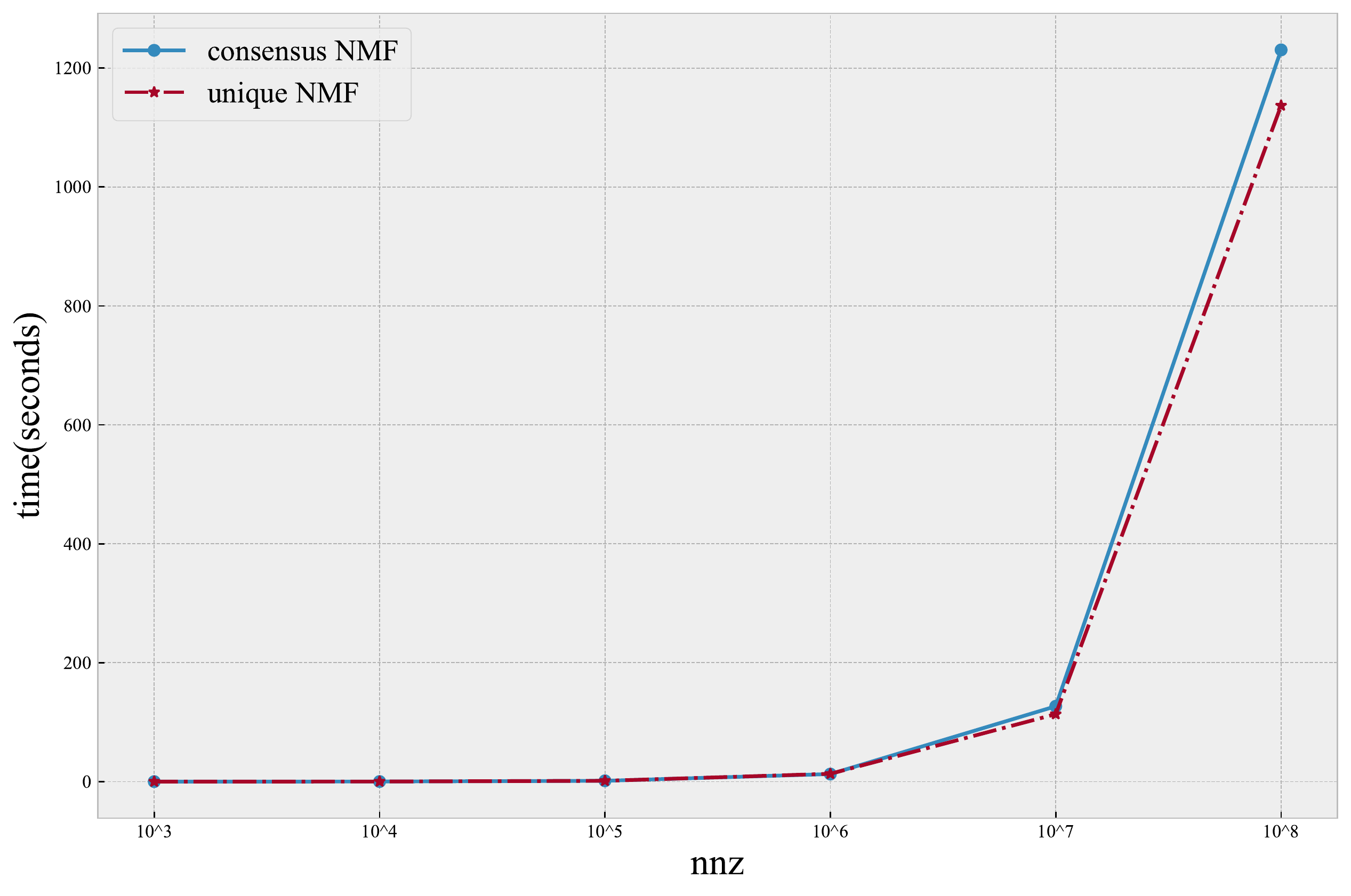}
    \caption{Running time on synthetic data}
    \label{time_cost}
\end{figure}

\subsection{Analysis of Running Time}
As we discussed in Sec.~\ref{complexity}, the time complexity of MNMF is positively related to the number of nonzero entries~(nnz) of the factorized matrices. In order to examine the time cost of the MNMF, we utilize the synthetic data by controlling nnz artificially to test the running time and the results are shown in Fig.~\ref{time_cost}. The "consensus NMF" means the running time of 50 iterations of consensus NMF process per stage and the "unique NMF" means the running time of 50 iterations of unique NMF process per stage, which is also a standard NMF process. As we can see, the consensus NMF process almost has a same running time as unique NMF process when the data size is relatively small. When we test on a large size data~(nnz equals $10^8$), the running time of the consensus NMF process is also close to standard NMF process. To some extent, our proposed consensus NMF process, existing various changes to the standard NMF process, does not significantly increase the complexity of it.

\section{Conclusion and Future Work} \label{conclusion}
In this paper, we try to explore the heterogeneous edges for network embedding. By considering each type of relationship as a view, we formalize the task as a multi-view network embedding problem and propose a novel MNMF model to solve it. More specifically, our MNMF model utilizes multi-stage strategy to preserve the consensus and unique information in multiple views. Meanwhile, the multi-stage matrix factorization method also serves the purpose for reducing the approximation error. We further demonstrate the MNMF model is a gradient boosting-like matrix factorization approach, thus the reason of the superiority of the MNMF model is also illustrated. To solve the optimization problem of MNMF effectively, we have introduced an coordinate gradient descent method via updating each variable alternatively. Experimental results not only indicate the effectiveness of the proposed model but also investigate the contributions of different modules.

Despite our model has achieved satisfactory results on representing networks with heterogeneous edges, there are still much room for future work. First, we will apply our model to more realistic scenarios, such as recommendation systems, community detection, personas and so on. Second, instead of tuning the hyper-parameter manually, we want to figure out the methods that automatically determine the stages responsible for preserving consensus information or preserving unique information. Moreover, a meaningful direction is extending the existing MNMF model to represent general heterogeneous networks, that is, the nodes and edges of the networks have multiple types at the same time, which is a challenging but highly significant problem.

\section*{Acknowledgement}
The research was supported by National Natural Science Foundation of China (No. 61802140).

\bibliographystyle{ACM-Reference-Format}
\bibliography{ref}

\appendix

\section*{Appendix: Optimization for MNMF model}
In this part, we give the details of the optimization process for MNMF model.
Since the objective function in Eq.~\eqref{consensus} is not convex over all variables simultaneously \cite{lee1999learning}, we optimize them by successively updating only one variable at a time, while fixing the other variables. In order to update $\bm{U}^l$ and $\bm{V}_k^l$, according to matrix property $\mathrm{Tr}(\bm{AB}) = \mathrm{Tr}(\bm{BA})$ and $\mathrm{Tr}(\bm{A}^\mathrm{T}) = \mathrm{Tr}(\bm{A})$, the objective function for Eq.~\eqref{consensus} can be rewritten as:
\begin{equation}
\begin{aligned}
    \mathcal{O} &= \sum_{k=1}^K (\alpha_k^l)^{\gamma} \mathrm{Tr}((\bm{R}_k^l-\bm{U}^l\bm{V}_k^l)^\mathrm{T}(\bm{R}_k^l-\bm{U}^l\bm{V}_k^l)) \\
    &= \sum_{k=1}^K (\alpha_k^l)^{\gamma} [\mathrm{Tr}(\bm{R}_k^l{\bm{R}_k^l}^\mathrm{T}) -2\mathrm{Tr}(\bm{R}_k^l{\bm{V}_k^l}^\mathrm{T}{\bm{U}^l}^\mathrm{T}) \\
    &+ \mathrm{Tr}(\bm{U}^l\bm{V}_k^l{\bm{V}_k^l}^\mathrm{T}{\bm{U}^l}^\mathrm{T})]
\end{aligned}
\label{obj4consensus}
\end{equation}
Here $\mathrm{Tr}(\bm{X})$ means the trace of matrix $\bm{X}$. Let $\bm{\Psi}^l$ and $\bm{\Theta}_k^l$ be the Lagrangian multiplier matrices for the constraint $\bm{U}^l \geq 0$ and $\bm{V}_k^l \geq 0$. We introduce Lagrange function as:
\begin{equation}
    \begin{aligned}
        \mathcal{L} &=  \mathcal{O} + \mathrm{Tr}(\bm{\Psi}^l{\bm{U}^l}^\mathrm{T}) + \sum_{k=1}^K \mathrm{Tr}(\bm{\Theta}_k^l{\bm{V}_k^l}^\mathrm{T})
    \end{aligned}
\label{lag4consensus}
\end{equation}
The partial derivatives for $\mathcal{L}$ with respect to $\bm{U}^l$ and $\bm{V}_k^l$ are: 
\begin{equation}
    \begin{aligned}
        \frac{\delta \mathcal{L}}{\delta \bm{U}^l} &= \sum_{k=1}^K (\alpha_k^l)^{\gamma}(-2\bm{R}_k^l{\bm{V}_k^l}^\mathrm{T} + 2\bm{U}^l\bm{V}_k^l{\bm{V}_k^l}^\mathrm{T}) + \bm{\Psi}^l \\
        \frac{\delta \mathcal{L}}{\delta \bm{V}_k^l} &= -2(\alpha_k^l)^{\gamma} {\bm{U}^l}^\mathrm{T} \bm{R}_k^l + 2 (\alpha_k^l)^{\gamma} {\bm{U}^l}^\mathrm{T}\bm{U}^l\bm{V}_k^l + \Theta_k^l
    \end{aligned}
    \label{de4consensus}
\end{equation}
Using the KKT conditions \cite{boyd2004convex} ${\bm{\Psi}^l}\odot{\bm{U}^l} = 0$ and ${\bm{\Theta}_k^l}\odot{\bm{V}_k^l} = 0$, we get the following equations:
\begin{equation}
    \begin{aligned}
        (\sum_{k=1}^K (\alpha_k^l)^{\gamma}(-\bm{R}_k^l{\bm{V}_k^l}^T + \bm{U}^l\bm{V}_k^l{\bm{V}_k^l}^T))\odot \bm{U}^l = 0 
    \end{aligned}
    \label{eq4consensus_U}
\end{equation}
\begin{equation}
    (-{\bm{U}^l}^\mathrm{T} \bm{R}_k^l + {\bm{U}^l}^\mathrm{T}\bm{U}^l\bm{V}_k^l)\odot \bm{V}_k^l = 0
    \label{eq4consensus_V}
\end{equation}
Then we drive the update rules as:
\begin{equation}
    \begin{aligned}
        \bm{U}^l &\leftarrow \bm{U}^l \odot \frac{\sum_k^K (\alpha_k^l)^{\gamma} \bm{R}_k^l {\bm{V}_k^l}^\mathrm{T}}{\sum_k^K (\alpha_k^l)^{\gamma} \bm{U}^l \bm{V}_k^l {\bm{V}_k^l}^\mathrm{T}} \\
        \bm{V}_k^l &\leftarrow \bm{V}_k^l \odot \frac{{\bm{U}^l}^\mathrm{T} \bm{R}_k^l}{{\bm{U}^l}^\mathrm{T} \bm{U}^l \bm{V}_k^l}
    \end{aligned}
    \label{solution_UV}
\end{equation}
Here $\odot$ and $\frac{[\cdot]}{[\cdot]}$ are element-wise multiplication and division respectively.

\begin{theorem}\label{th1}
The limited solutions of the updating rules in Eq.~\eqref{solution_UV} satisfy the KKT optimality condition.
\end{theorem}
\begin{proof}
At convergence, we have ${\bm{U}^l}^{(\infty)} = {\bm{U}^l}^{(t+1)} = {\bm{U}^l}^{(t)} = \bm{U}^l$,  where $t$ denotes the $t$-th iteration. That is,
\begin{equation}
    \bm{U}^l = \bm{U}^l \odot \frac{\sum_k^K (\alpha_k^l)^{\gamma} \bm{R}_k^l {\bm{V}_k^l}^\mathrm{T}}{\sum_k^K (\alpha_k^l)^{\gamma} \bm{U}^l \bm{V}_k^l {\bm{V}_k^l}^\mathrm{T}}
\end{equation}
which is equivalent to 
\begin{equation}
    (\sum_{k=1}^K (\alpha_k^l)^{\gamma}(-\bm{R}_k^l{\bm{V}_k^l}^\mathrm{T} + \bm{U}^l\bm{V}_k^l{\bm{V}_k^l}^\mathrm{T}))\odot \bm{U}^l = 0
\label{kkt}
\end{equation}
Clearly, Eq.~\eqref{kkt} is identical to Eq.~\eqref{eq4consensus_U}. In the same way, the correctness of the updating rules in Eq.~\eqref{update4consensus} for $\bm{V}_k^l$ can be proved.
\end{proof}

We follow the similar method in \cite{cai2013multi} to update $\alpha_k^l$ in Eq.\eqref{consensus}. Let's denote $\bm{W}_k^l = {\Vert \bm{R}_k^l - \bm{U}^l\bm{V}_k^l \Vert_F^2}$, then Eq.~\eqref{consensus} is reformulated as
\begin{equation}
    \begin{aligned}
    \text{min} &\qquad \sum_{k=1}^K (\alpha_k^l)^{\gamma} \bm{W}_k^l,\\
    \text{s.t.} &\qquad \sum_{k=1}^K \alpha_k^l=1, \alpha_k^l \geq 0 \\
    \end{aligned}
    \label{consensus_r}
\end{equation}
The Lagrange function of Eq.~\eqref{consensus_r} is derived as:
\begin{equation}
\centering
    \text{min} \qquad \sum_{k=1}^K (\alpha_k^l)^{\gamma} \bm{W}_k^l - \lambda(\sum_{k=1}^K \alpha_k^l - 1)
\label{lag}
\end{equation}
Here $\lambda$ is the Lagrangian multiplier. By taking the derivative of Eq.~\eqref{lag} with respect to $\alpha_k^l$ as zero, we get
\begin{equation}
    \centering
    \alpha_k^l = \left( \frac{\lambda}{\gamma \bm{W}_k^l} \right)^{\frac{1}{\gamma - 1}}
    \label{lagder}
\end{equation}
After that we replace $\alpha_k^l$ in Eq.~\eqref{lagder} into $\sum_{k=1}^K \alpha_k^l = 1$ to obtain update rule for $\alpha_k^l$:
\begin{equation}
    \centering
    \alpha_k^l \leftarrow \frac{\left(\gamma \bm{W}_k^l\right)^{\frac{1}{1-\gamma}}}{\sum_{k=1}^K\left(\gamma \bm{W}_k^l\right)^{\frac{1}{1-\gamma}}}
\end{equation}

\end{document}